\documentclass[preprint,10pt,3p]{elsarticle}

\usepackage{amssymb}
\usepackage{amsthm}

\usepackage{latexsym}
\usepackage{amsmath}
\usepackage{proof}

\usepackage{bbm}
\usepackage{interval}
\allowdisplaybreaks

\usepackage[all]{xy}

\usepackage{amsthm}

\usepackage{array,cmll}

\newtheorem{thm}{\bf Theorem}[section]
\newtheorem{exam}[thm]{\bf Example}
\newtheorem{prop}[thm]{\bf Proposition}
\newtheorem{lem}[thm]{\bf Lemma}
\newtheorem{cor}[thm]{Corollary}
\newtheorem{defn}[thm]{\bf Definition}

\newtheorem{rem}[thm]{\bf Remark}

\newtheorem{notn}[thm]{Notation}

\newcommand{\Comp}{\circ}

\newcommand{\SRel}{\sf{SRel}}

\newcommand{\TKer}{{\sf TK}}

\newcommand{\Pcoh}{\sf{Pcoh}}
\newcommand{\Coh}{\sf{Coh}}

\newcommand{\Po}[1]{\sf{P} #1}

\newcommand{\ms}[1]{(#1, \mathcal{#1})}

\newcommand{\msi}[2]{(#1_{#2}, \mathcal{#1}_{#2})}

\newcommand{\zeroinf}{\overline{\mathbb{R}}_{+}}

\newcommand{\pr}{{\sf p}}

\newcommand{\msbf}[1]{\mathbbm{#1}}

\newcommand{\abs}[1]{\mid \!{#1} \! \mid}

\newcommand{\mc}[1]{\mathcal{#1}}

\newcommand{\ort}[3]{#1 \, \, \bot_{#2} \, \, \, #3}

\newcommand{\inpro}[3]{\langle #1  \! \mid  \! #3 \rangle_{#2}}

\newcommand{\inprocoh}[2]{\langle #1  , #2 \rangle}

\newcommand{\borel}{\mc{B}_{+}}

\newcommand{\FI}{F^{\mbox{-}1}}

\newcommand{\opTKer}{{\sf TKer}^{\mbox{\tiny \sf op}}}
\newcommand{\opTKersfin}{{\sf TsKer}^{\mbox{\tiny \sf op}}\
}

\newcommand{\TsKer}{{\sf TsK}}

\newcommand{\TsKeromg}{\TsKer_{\omega}}

\newcommand{\opTsKer}{\TsKer^{\mbox{\tiny \sf op}}}

\newcommand{\cC}{\mathcal{C}}

\newcommand{\oble}[2]{#1^{\leq #2}}
\newcommand{\obwt}[2]{(#1 \& I)^{\otimes #2}}
\newcommand{\gobwt}[2]{(#1 \& \msbf{I})^{\otimes #2}}

\newcommand{\crc}[1]{#1^{\circ}}

\newcommand{\ccrc}[1]{#1^{\circ \circ}}

\newcommand{\cccrc}[1]{#1^{\circ \circ \circ}}

\newcommand{\T}[2]{{\bf O}_#1(#2)}

\newcommand{\eq}{{\sf eq}}

\newcommand{\opTsKeromg}{\TsKeromg^{\mbox{\tiny \sf op}}}

\newcommand{\msn}[2]{(#1^{(#2)}, \mathcal{#1}^{(#2)})}

\newcommand{\msnw}[2]{(\mc{#1} \& \mc{I})^{(#2)}}

\usepackage{relsize}

\usepackage{setspace}

\newcommand{\Dd}[2]{\delta(#2, \,#1)}

\newcommand{\Int}[3]{\int_{#1} #2 \, \, \, #3}

\newcommand{\bang}[1]{! #1}

\newcommand{\bra}[2]{\eq \backslash s_n(\otimes_{i=1}^#2 
\bar{#1}_i)}

\newcommand{\dbra}[2]{ \langle \!  \langle #1_1, \ldots , #1_{#2}
\rangle \!  \rangle}

\newcommand{\setone}[1]{[#1]}

\newcommand{\Sub}{{\bf O}^{\#}_{\mc{I}}(\opTsKer)}
\newcommand{\Subomg}{{\bf O}^{\#}_{\mc{I}}(\opTsKeromg)}
\newcommand{\C}[2]{{\sf c}_{#1} (#2)}

\newcommand{\crceq}[1]{#1^{\circ \eq}}

\pdfoutput=1 
\begin{document}

\begin{frontmatter}



\title{ Double Glueing over Free Exponential:
with Measure Theoretic Applications}



\author{Masahiro HAMANO}
\ead{hamano@gs.ncku.edu.tw, m-hamano@nagasaki-u.ac.jp}
\address{
Nagasaki University, \\ 
1-14 Bunkyo, Nagasaki City 852-8521,
JAPAN}

\begin{abstract}
This paper provides a compact method to lift
the free exponential construction of Melli\`{e}s-Tabareau-Tasson
over the Hyland-Schalk double glueing for orthogonality categories.
A condition "reciprocity of orthogonality"
is shown simply enough to lift the free exponential
over the double glueing in terms of  the orthogonality.
Our general method applies to the monoidal category $\TsKer$ 
of the s-finite transition kernels with countable biproducts.
We show
(i) $\opTsKer$
 has the free exponential, which is shown to be
 describable
in terms of measure theory.
(ii) The s-finite transition kernels have an orthogonality between
measures and measurable functions in terms of Lebesgue integrals.
The orthogonality has the reciprocity, hence the free exponential of (i)
lifts to the orthogonality category 
$\T{\mc{I}}{\opTsKer}$, 
which
subsumes Ehrhard et al's probabilistic
coherent spaces as a full subcategory of countable measurable
spaces. 
To lift the free exponential,
the measure-theoretic uniform convergence theorem commuting Lebesgue integral and limit
 plays a crucial role as well as Fubini-Tonelli theorem for double integral
in s-finiteness. 
Our measure-theoretic orthogonality is considered
as a continuous version of the orthogonality of the probabilistic
coherent spaces for linear logic, and in particular provides a
two layered decomposition of Crubill\'{e} et al's direct
free exponential for these spaces arisen as discretisation in this paper.
\end{abstract}







\end{frontmatter}



\begin{spacing}{0.1}
\tableofcontents
\end{spacing}

\section*{\large Introduction}
This paper is concerned with modelling the exponential connective of linear logic;
(i) abstractly for the orthogonality for the double glueing construction
and (ii) concretely for a category of s-finite
transition kernels using the duality between measures and measurable
functions. 

\smallskip

The (symmetric) monoidal category provides a minimum 
categorical counterpart to the tensor connective of
linear logic \cite{GirLL}. On top of the monoidality,
richer categorical structures are augmented consistently,
interpreting other logical components
(e.g., the closedness for the linear implication and Barr's *-autonomy
for the duality of linear logic).
Understanding categorical properties of the exponential connective
! of linear logic is the difficult part in various works
(cf. \cite{BS, HSha, Mellies} for surveys) stemming from Seely
\cite{RS}.
Recently Melli\`{e}s-Tabareau-Tasson \cite{MTT} formulates the
categorical construction to obtain the free commutative comonoid over a
symmetric monoidal category with binary products.
Their construction interprets the exponential as
the limit of the enumerated equalisers for $n!$-symmetries of $n$-th
powers of the
monoidal products between certain rooted objects using the cartesian product.

The most recent application of \cite{MTT}
is done by Crubill\'{e}-Ehrhard-Pagani-Tasson \cite{Crubille}
to Danos-Ehrhard \cite{DE}'s probabilistic coherent spaces,
whose exponential is shown to be the free one.
The category $\Pcoh$
of probabilistic coherent spaces is a probabilistic version
of Girard's denotational semantics $\Coh$ of coherent spaces \cite{GirLL}.
As the original semantics of linear logic,
$\Coh$ has the distinctive feature of the linear duality,
in terms of the graphical structure on webs,
stating that a clique and anti-clique intersect in at most one
singleton.
Developing the web based method,
Ehrhard investigates the linear duality in
mathematically richer structures (e.g., K\"{o}the spaces \cite{EhrKothe}),
and his investigation
leads Danos-Ehrhard \cite{DE} to a probabilistic version of the duality,
generalising the web to the non-negative real valued functions on it
(i.e., fuzzy web),
reminiscent of probability distributions on the web.
When each function is identified as a vector enumerating its values,
the probabilistic linear duality states that
the inner product of the vectors from clique and anti-clique
is not greater than $1$.
Girard earlier addresses this quantitative form of
the duality in \cite{Girquan}.

The starting point of this paper 
is our attempt
to comprehend \cite{Crubille}'s application to probabilistic
semantics more generally, especially free from the web-based method,
but extracting the abstract role of the duality more explicitly.
We present two new methods respectively to
the abstract orthogonality 
in the category-theory and to
probabilistic semantics in the continuous measure-theory:
(i) Lifting a free exponential to an intricate category with reciprocity of
the orthogonality.
(ii) Continuous orthogonality for linear logic
   between measures and measurable functions.

For (i), we investigate the Hyland-Schalk double glueing construction
\cite{HSha}. 
Our construction is general enough to yield
a probabilistic version of the 
linear duality.
It is well known
that the double glueing construction
over the category of relations gives rise to the *-autonomy of 
$\Coh$ \cite{BS,HSha}. 
This leads to various
full completeness theorems,
not only for the multiplicative 
\cite{DHPP} fragment, but also for the multiplicative-additive 
\cite{BHS} one. We start with observing {\em focused orthogonality}
characterises the adjunction of the orthogonality.
We show that this simple notion of orthogonality
interact consistently with equalisers and limit of \cite{MTT}
so to lift the free exponential to the orthogonality categories.
Our construction gives a simple insight
for the less studied exponential structure inside the double glueing
after \cite{HSha}. 
Importantly, the insight reveals the two layers decomposition of 
the exponential (base category level and its double glueing lifting)
in terms of Melli\`{e}s-Tabareau-Tasson construction.
 \\
$\quad$ For (ii), we investigate the s-finite (i.e., sum of finite) class of
the transition kernels \cite{Sharpe}, which is recently revisited by Staton
\cite{Stat} to analyse functorial commutativity of measure-based denotational
semantics for probabilistic programming.
The s-finiteness provides a wider extension of
the preceding probabilistic semantics using  
transition kernels, from Kozen's precursory work \cite{Koz},
then Panangaden's seminal work
of Markov kernels (a.k.a, stochastic relations)
\cite{Prapaper,Prabook} to the recent measure-transformer semantics
\cite{Borg} of finite kernels.
Following Staton's work on the functorial monoidal product,
we show another advantage (versus Markov and finiteness):
s-finiteness is wide enough
to accommodate the free exponential of 
Melli\`{e}s-Tabareau-Tasson.

\smallskip

We note an unresolved issue of the paper. Both the s-finite transition kernels and their double glueing lifting lack a closed structure for monoidal products within the continuous framework, rendering them inconclusive as a complete model of linear logic.

\smallskip

The paper is organised as follows: 
Section \ref{sect1} is a categorical study
of how to lift a Melli\`{e}s-Tabareau-Tasson
free exponential of $\cC$ to a double glueing $\T{J}{\cC}$
with a focused
orthogonality. 
 Section \ref{sect2} is a measure-theoretic study on
the s-finite transition kernels $\TsKer$.
Section \ref{sect3} constructs the free exponential in $\opTsKer$.
Section \ref{sect4} presents a measure theoretic
instance of Section \ref{sect1} using Section \ref{sect3},
which subsumes the probabilistic coherence spaces as discretisation.


\section{\large Lifting Free Exponential of $\cC$
to $\T{J}{\cC}$} \label{sect1}
This section concerns lifting of
Melli\`{e}s-Tabareau-Tasson free exponential
of $\cC$ to an orthogonality category $\T{J}{\cC}$
by reciprocity for focused orthogonality.
Melli\`{e}s-Tabareau-Tasson free exponential construction
consists of four conditions; the equalisers, the limit of the equalisers 
and the distribution of the monoidal product over them.
While the first two conditions on the equalisers and on
their limit are automatically lifted to $\T{J}{\cC}$,
necessary and sufficient conditions are formulated
in terms of orthogonality
for the remaining two conditions on distributivity of the monoidal product of $\T{J}{\cC}$.

In what follows throughout the paper,
the identity morphism on object $X$ of a category is simply denoted by $X$.

\subsection{\normalsize Melli\`{e}s-Tabareau-Tasson Free Exponential in a Monoidal
$\cC$ with Finite Products}
\begin{defn}[Melli\`{e}s-Tabareau-Tasson free exponential
 \cite{MTT}]\label{MTTfex}{\em
The following four structures
uniquely determine the free exponential of 
a symmetric monoidal category $(\cC, \otimes, I)$
with finite cartesian products $\&$:
\begin{description}
\item[(E$_A$)] For any object $A$ of $\cC$ and a natural number $n$,
the equaliser exists in $\cC$, denoted by $\oble{A}{n}$ with $\eq_A$,
for the $n!$-symmetries on $\obwt{A}{n}$ as the parallel morphisms: 
$$\xymatrix{
\oble{A}{n} \ar[r]^{\eq_A} &
\obwt{A}{n} \ar@<-7pt>[r]_{\mbox{\tiny $n!$ symm.}}^{\vdots}  \ar@<7pt>[r]
&  \obwt{A}{n} \\
\ar@{-->}[u]^{\exists_! \, \, \eq \backslash f}
\phantom{ttt} \ar[ru]_f & 
  }$$
For $f: X \longrightarrow \obwt{A}{n}$ equalising the $n!$-symmetries,
the universal morphisms factoring $f$ is denoted by $\eq\backslash f$
such that $\eq\Comp (\eq\backslash f)=f$.

\item[(distribution of $\otimes$ over E$_A$)] The equaliser
	   of E$_A$
commutes with the tensor product: \\
$(\oble{A}{n} \otimes B, \eq\otimes B)$
becomes the equaliser for the $n!$-symmetries $\otimes B$
on $\obwt{A}{n} \otimes B$ for any object $B$: 
$$\xymatrix{
\oble{A}{n} \otimes B \ar[r]^{\eq\otimes B} &
\obwt{A}{n} \otimes B  \ar@<-7pt>[r]_{\mbox{\tiny ($n!$ symm.) $\otimes B$}}^{\vdots}  \ar@<6pt>[r]
&  \obwt{A}{n} \otimes B 
  }$$

\item[(L$_{A}$)] For any object $A$ in $\cC$,
the following diagram 
has the limit $(\oble{A}{\infty},
\{ p_{\infty,n}: \oble{A}{\infty} \longrightarrow \oble{A}{n} \}_n )$:
$$
\xymatrix
@C=20pt
{
\oble{A}{0} &    \oble{A}{1} \ar[l]_{p_{1,0}} &  
\oble{A}{2}  \ar[l]_{p_{2,1}} \cdots \cdots
& \ar[l]
\oble{A}{n} &  \oble{A}{n+1} \ar[l]_{p_{n+1,n}}  \cdots
}
$$
where $p_{n+1,n}$ is the universal morphism guaranteed by E$_A$
for the composition $(\obwt{A}{n} \otimes \, \pr_r) \Comp \eq$
equalising $n!$-symmetries $\obwt{A}{n}$. 
$\pr_r$ is the right projection. See the following:
$$\xymatrix{
\oble{A}{n+1} \ar[d]^(.4){\eq} 
\ar[rr]^{ p_{n+1,n}} & &  \ar[d]_{\eq} \oble{A}{n}
\\
 \obwt{A}{n+1} \ar[rr]^{
\obwt{A}{n} \otimes \, \pr_r} &  & \obwt{A}{n} \otimes I 
\cong  \obwt{A}{n} }$$

 \item[(distribution of $\otimes$ over  L$_A$)] The limit of L$_A$
	    commutes with the monoidal product $\otimes$: \\
$(\oble{A}{\infty} \otimes B, \{ p_{\infty, n} \otimes B \}_n )$
becomes the limit for the following diagram for any object $B$:
$$
\xymatrix
@C=20pt
{
\oble{A}{0} \otimes B 
& \ar@<-2pt>[l]_{p_{1,0} \otimes B} \cdots \cdots & 
\ar[l] \oble{A}{n} \otimes B &  \oble{A}{n+1} \otimes B \ar@<-2pt>[l]_(.55){p_{n+1,n} \otimes B}  \cdots
}
$$
\end{description}
}
\end{defn}
The constructions of the equaliser for (E$_A$) and
the limit for (L$_A$)
act not only on objects but also on morphisms functorially
preserving the categorical composition.
\begin{defn}[morphisms $\oble{f}{n}$ and $\oble{f}{\infty}$] \label{actmor}{\em
Let $f: A \longrightarrow B$. 
\begin{itemize}
 \item 
The condition E$_A$ guarantees that  
there exists the unique morphism $\oble{f}{n}$ for any natural number $n$: 
\begin{align*}
\oble{f}{n} := \eq_B \backslash
\big(  (f \cdot \pr_l \, \,  \& \,\, I \cdot \pr_r)^{\otimes n} 
\Comp \eq_A \big): 
\oble{A}{n} \longrightarrow \oble{B}{n}
\end{align*}
because the composition
$(f \cdot \pr_l \& I \cdot \pr_r)^{\otimes n}
\Comp \eq_A$ equalises the $n!$-symmetries $\obwt{B}{n}$. 
See the following diagram,
where $\pr_l$ and $\pr_r$ are right and left projections:
$$\xymatrix{
\oble{B}{n} \ar[r]^{\eq_B}
& \obwt{B}{n} \\
\oble{A}{n} \ar[u]^{\oble{f}{n}}\ar[r]_{\eq_A}
& \obwt{A}{n} 
\ar[u]_{(f \cdot \pr_l  \, \& \, I \cdot \pr_r)^{\otimes n}}
}$$
From the universality of $\eq_B$,
it holds that $\oble{f}{n}
=\oble{g}{n}$ whenever $(f \cdot \pr_l \, \, \& \, \,  I \cdot \pr_r)^{\otimes n}=
(g \cdot \pr_l \, \,  \& \, \,  I \cdot \pr_r)^{\otimes n}$.

\item
The condition L$_B$ guarantees that 
there exists the unique morphism
 \begin{align*}
\oble{f}{\infty} :  \oble{A}{\infty} \longrightarrow \oble{B}{\infty}
\end{align*}
factoring the cone $\{
\xymatrix{\oble{A}{\infty} \ar[r]^{p_{\infty n}}
 &   \oble{A}{n} \ar[r]^{\oble{f}{n}} & \oble{B}{n}
} \}_n$.
 See the following diagram for the universal
$\oble{f}{\infty}$ for the cone.
$$
\xymatrix@R=15pt
{
\cdots
&
\ar[l]
\oble{B}{n} &  \oble{B}{n+1} \ar[l]_{p_{n+1,n}}  
& \cdots &  \oble{B}{\infty}
 \\
\cdots
& \ar[l]
\oble{A}{n} \ar[u]^{\oble{f}{n}} &  \oble{A}{n+1} \ar[l]_{p_{n+1,n}} 
\ar[u]^{\oble{f}{(n+1)}} 
& \cdots & \oble{A}{\infty} \ar@{-->}[u]_{\exists_{\, !} \, \, \oble{f}{\infty}}
\ar@/^10pt/[ll]^(.7){p_{\infty, n+1}}
\ar@/^20pt/[lll]^(.8){p_{\infty, n}}
}
$$
The indexed morphism indeed is a 
     cone by the equality (\ref{square}) below, whose demonstration
is put in \ref{apsquare}.
\begin{align} \label{square}
p_{n+1,n} \Comp \oble{f}{n+1} = 
\oble{f}{n} \Comp p_{n+1,n}
\end{align}
\end{itemize}
}\end{defn}

\noindent Obviously the two morphisms in Definition \ref{actmor}
are related $\begin{aligned}
\oble{f}{n} \Comp p_{\infty, n} = p_{\infty, n} \Comp
\oble{f}{\infty}
\end{aligned}$, which is seen
as the limit of (\ref{square}).

\subsection{\normalsize Focused Orthogonality in $\cC$ and Orthogonality
Category $\T{J}{\cC}$} \label{subsect1.2}
\begin{defn}[orthogonality on $\cC$ \cite{HSha}]\label{defort}{\em
An {\em orthogonality} on a symmetric monoidal category $(\cC, \otimes, I)$
is an indexed family of relation $\bot_R$ between the maps
$u \in \cC(I,R)$ and $x \in \cC(R,J)$,
where $I$ is the monoidal unit while
$J$ is an arbitrary fixed object, 
\begin{align*}
\ort{I \stackrel{u}{\longrightarrow} R}{R}{
R \stackrel{x}{\longrightarrow} J}
\end{align*}
satisfying the following conditions: 

\noindent ({isomorphism})
If $f: R \rightarrow S$ is an isomorphism, then 
for any $u: I \longrightarrow R$ and $x: R \rightarrow J$,
\begin{align*}
\ort{u}{R}{x} \quad \txt{iff} \quad \ort{f \Comp u}{S}{x \Comp f^{-1}}
\end{align*}

\noindent ({tensor})
Given $u: I \longrightarrow R$,
$v: I \longrightarrow S$, and $h: R \otimes S \longrightarrow J$,
\begin{align*}
& \ort{u}{R}{R \cong R \otimes I \stackrel{R \otimes v}{\longrightarrow} R \otimes S
 \stackrel{h}{\longrightarrow} J} \, \txt{and}  \\
& 
\ort{v}{S}{S \cong I \otimes S \stackrel{u \otimes S}{\longrightarrow}
 R \otimes S \stackrel{h}{\longrightarrow} J} \quad \txt{
imply}  \quad \ort{u \otimes v}{R \otimes S}{h}.
\end{align*}

\noindent ({identity})
For all $u: I \rightarrow R$ and $x: R \rightarrow J$,
$\ort{u}{R}{x}$ implies $\ort{\operatorname{Id}_I}{I}{x \Comp u}$

}\end{defn}

For $U \subseteq \cC(I,R)$, its orthogonal $\crc{U} \subseteq \cC(R,J)$
is given by 
\begin{align*}
\crc{U}:= \{ x :  R \longrightarrow J \mid \forall u \in U \, \,  \ort{u}{R}{x}\} 
\end{align*}
Throughout the paper, $\ort{x}{}{U}$ denotes a short for 
$x \in \crc{U}$.

This gives a Galois connection so that $\crc{U} \subset V$ iff
$U \subset \crc{V}$ for any $U$ and $V$. The closure operator $\ccrc{(-)}$
fixes $\crc{U}$ so that
$U^{\circ \circ \circ } =
\crc{U}$.
In terms of orthogonality, $\ort{y}{}{V}$ iff $\ort{y}{}{\ccrc{V}}$
for all $y$.

\smallskip

In this paper, a special kind of orthogonality is considered, 
introduced by Hyland-Schalk \cite{HSha} originally 
in order to define a certain class $F$ of morphisms, called focused.
\begin{defn}[focused orthogonality (Example 48 of \cite{HSha})]
{\em
An indexed family of relation $\bot_R$ is {\em focused}
when it is determined by a subset $F$ of $\cC(I,J)$
in the following manner:
\begin{align}
\ort{
I \stackrel{u}{\longrightarrow} R}{R}
{R \stackrel{x}{\longrightarrow} J} \quad \quad
& \mbox{if and only if} \quad \quad
x \Comp u \in F \label{focort}
\end{align}
}\end{defn} 
This stipulates that the orthogonality is {\em reciprocal} since
(\ref{focort}) is alternatively characterised as follows
for every $u: I \longrightarrow R$,
$x : S \longrightarrow J$, and $f: R \longrightarrow S$,
\begin{align}
  \ort{u}{R}{x \Comp f} \quad \quad
& \mbox{if and only if} \quad \quad  \ort{f \Comp u}{S}{x} \label{exadiort}
\end{align}
In what follows in this paper,
the reciprocity (\ref{exadiort}) of orthogonality is used 
to characterise the focused orthogonality accordingly to the following
equivalence:
\begin{lem}[reciprocity coincides with focused orthogonality]
(\ref{exadiort}) if and only if 
(\ref{focort})
\end{lem}
\begin{proof}
\noindent (if)
The left and the right of (\ref{exadiort}) are both equivalent to 
$x \Comp f \Comp u \in F$.  \\
\noindent (only if) 
Put $F := \{ x:I \longrightarrow J \mid \ort{\operatorname{Id}_I}{I}{x} \}$.
Then $\ort{u}{R}{x}$
iff $x \Comp u \in F$ by the
reciprocity.
\end{proof}

\begin{defn}{\em
{\em Precise tensor} is the tensor condition of Definition \ref{defort} 
strengthened by replacing ``imply'' with ``iff''}
\end{defn}

\begin{prop}[reciprocity enough for the three conditions of the
orthogonality] The reciprocity of the focused orthogonality
derives the three conditions (isomorphism), (tensor)
and (identity) on Definition \ref{defort}.
Moreover (precise tensor) is derived.
\end{prop}
\begin{proof}
\noindent (isomorphism) $\ort{u=f^{-1} \Comp f \Comp u}{R}{x}$
iff $\ort{f \Comp u}{R}{x \Comp f^{-1}}$ by reciprocity.

\noindent (identity) Similarly but easily by $u= u \Comp \operatorname{Id}_I$.

\noindent (tensor) Only one premise of the tensor condition implies
the conclusion: $\ort{u}{R \otimes S}{h \Comp (R \otimes v)}$
iff by reciprocity $\ort{(R  \otimes v) \Comp u}{R \otimes S}{h}$,
whose left hand is $u \otimes v$.

\noindent (precise tensor)
$u \otimes v = (R \otimes v) \Comp (u \otimes I)
=  (u \otimes S) \Comp (I \otimes v)$,
composing $(R \otimes v)$
(resp. $(u \otimes S)$) to the right $h$ in
the $\ort{}{R \otimes S}{}$ 
yields the first (resp. the second) premise
of the tensor.
\end{proof}


\smallskip

\begin{defn}[orthogonality category $\T{J}{\cC}$ \cite{HSha}] {\em
Let us fix an orthogonal relation. An object of $\T{J}{\cC}$ is a tuple $\msbf{A}=(A, \msbf{A}_p, \msbf{A}_{cp} )$ with  
$\msbf{A}_p \subseteq{\cC}(I,A)$
and $\msbf{A}_{cp} \subseteq {\cC}(A,J)$ satisfying; \\

\noindent 
(mutual orthogonality)
$\begin{aligned}
\msbf{A}_p=\crc{(\msbf{A}_{cp})} \quad \txt{and} \quad 
\msbf{A}_{cp}=\crc{(\msbf{A}_p)}
\end{aligned}$ \\

\noindent Each map from $\msbf{A}=(A,\msbf{A}_p, \msbf{A}_{cp})$
to $\msbf{B}=(B, \msbf{B}_p, \msbf{B}_{cp})$ in $\T{J}{\cC}$ is any $\cC$
map $f: A \longrightarrow B$ satisfying: \\
\noindent (p) point:
$\forall u  : I \longrightarrow A$ in $\msbf{A}_p$,
the composition $f \Comp u$:  $\xymatrix{I \ar[r]^{u} & A \ar[r]^f & B 
}$ belongs to $\msbf{B}_p$. \\
\noindent (cp) copoint:
$\forall y : B \longrightarrow J$ in $\msbf{B}_{cp}$,
the composition $y \Comp f$: $\xymatrix{A \ar[r]^{f} & 
B \ar[r]^y & J 
}$ belongs to $\msbf{A}_{cp}$.
}
\end{defn}

The functor exists 
$\abs{\, } : \T{J}{\cC} \longrightarrow
\cC$ forgetting the second and the third components of the objects.

\begin{rem}[$\T{J}{\cC}$ is Hyland Schalk's tight orthogonality category]
{\em The category $\T{J}{\cC}$ is called the {\em tight orthogonality
category} in Hyland-Schalk \cite{HSha}, whereby it is formulated as a
subcategory of the double glueing category over $\cC$ (cf. Definition 47
of \cite{HSha}).
Since this is the only 
double glueing construction concerned in the present paper,
we use the simple name.}
\end{rem}

\begin{lem}\label{derpcp}
The conditions (p) and (cp)
are derivable from one another when an orthogonality is focused.
\end{lem}

\begin{proof}
By $ \msbf{B}_p= \crc{\msbf{B}_{cp}}$, the condition (p) says $\forall u
 \in \msbf{A}_p \, \forall s \in \msbf{B}_{cp}
 \, \, \ort{f \Comp u}{B}{s}$.
By $\msbf{A}_{cp}=\crc{\msbf{A}_p}$, the condition (cp) says $\forall s
 \in \msbf{B}_{cp} \, \forall u \in \msbf{A}_p \, \,  
 \ort{u}{A}{s} \Comp f$. The two are equivalent by the reciprocity.
\end{proof}
By Lemma \ref{derpcp}, when an orthogonality is focused,
an alternative definition of the category $\T{J}{\cC}$ is obtained;
\begin{defn}[$\T{J}{\cC}$ with a focused orthogonality]{\em 
When an orthogonality on $\cC$ is focused,
each object of $\T{J}{\cC}$ is represented alternatively by a pair
$\msbf{A}=(A, \msbf{A}_p)$ satisfying the following instead of 
the mutual orthogonality: 

\smallskip

\noindent 
(double orthogonality):
$\begin{aligned}
\ccrc{(\msbf{A}_p)} = \msbf{A}_p
\end{aligned}$ 

\smallskip

\noindent Each map between the objects
must satisfy the condition (p) only. 

A stronger condition suffices in particular
when the second components 
are represented by $\msbf{A}_p = \ccrc{U}$
and $\msbf{B}_p = \ccrc{V}$ with genuine subsets $U$
and $V$ of $\msbf{A}_p$
and of $\msbf{B}_p$ respectively:

\smallskip

\noindent ($\overline{\mbox{p}}$)
$\forall u \in U$,
the composition $f \Comp u$ belongs to $V$.
}\end{defn}
\smallskip
\noindent The sufficiency is because of the monotonicity of the operation
$\ccrc{(~)}$
and the following lemma. 
\begin{lem} \label{ccrclem}
For any morphism $f: A \longrightarrow B$ in $\cC$ and any subset $U \subseteq
 {\cC}(I,A)$, $f(\ccrc{U}) \subseteq \ccrc{f(U)}$. (See the proof in
\ref{apccrclem}).
\end{lem}

Both cartesian and monoidal products in $\cC$
are lifted to $\T{J}{\cC}$ respectively, as formulated in Section 5.3 of \cite{HSha}.
\begin{prop}[Product in $\T{J}{\cC}$ \cite{HSha}]\label{protjc}{\em 
Suppose $\cC$ has finite products
and an orthogonality is focused.
Then $\T{J}{\cC}$ has finite products
\begin{align*}
\msbf{A} \& \msbf{B} :=  \big( A \& B,
\msbf{A}_p \& \msbf{B}_p:= 
\{ u \& v \mid u \in \msbf{A}_p \, \,  v \in \msbf{B}_p   \} \big) 
\end{align*}
Note the second component is automatically closed under the double
 orthogonality.
The forgetful $\T{J}{\cC} \longrightarrow \cC$ preserves finite products.
The proof is put in \ref{approtjc}.}
\end{prop}

\begin{defn}[stable tensor (Definition 58 of \cite{HSha}] \label{stabten}{\em
An orthogonality on a monoidal category $\cC$
{\em stabilises} the monoidal product when the following condition holds
for all $U \subseteq \cC(I,R)$ and $V \subseteq \cC(I,S)$: \\

\noindent (stable tensor) 
$\begin{aligned} 
 \crc{(\ccrc{U} \otimes \ccrc{V})} =
\crc{(\ccrc{U} \otimes V)} = \crc{(U \otimes \ccrc{V})}
\end{aligned}$
}
 \end{defn}

The stable tensor is a condition on a representability
of certain maps in multicategories when $\cC$
has a closed structure on the monoidal product (see
Section 5.3 of \cite{HSha}). However the present paper
does not assume the closedness.

The focused
orthogonality is strong enough to stabilise the monoidal product in
$\cC$:
\begin{lem} \label{intstabten}
Any focused orthogonality stabilises monoidal products.
\end{lem}
\begin{proof}
We prove ($\supset$) of the stable tensor condition as the converse is
 tautological. Take any $\nu \in RHS$, which means
$\forall f \in \ccrc{U} \, \forall g \in V $ $\ort{f \otimes g
= (f \otimes S) \Comp (I \otimes g)}{R \otimes S}
{R \otimes S \stackrel{\nu}{\longrightarrow} J}$
iff by reciprocity $\ort{g}{S}
{S \cong I \otimes S \stackrel{f \otimes S}{\longrightarrow} R \otimes S
\stackrel{\nu}{\longrightarrow} J}$. But this means 
$\forall h \in \ccrc{V} \, \, \ort{h}{S}{\nu \Comp (f \otimes S)}$
iff by reciprocity
$\ort{f \otimes h=(f \otimes S) \Comp (I \otimes h)}{R \otimes S}{\nu}$, which means $\nu \in LHS$.
\end{proof}

\begin{defn}[Monoidal product in $\T{J}{\cC}$ \cite{HSha}] \label{stabmon} {\em 
Suppose $\cC$ is symmetric monoidal with an orthogonality
stabilising $\otimes$.
Then $\T{J}{\cC}$ is symmetric monoidal
and the forgetful $\T{J}{\cC} \longrightarrow \cC$ preserves the monoidality.
\begin{align*}
\msbf{A} \otimes \msbf{B} := \big( A \otimes B,
\ccrc{(\msbf{A}_p \otimes \msbf{B}_p)} \big) 
\end{align*}
The tensor unit $\msbf{I}$ is given by $(I, \ccrc{\{\operatorname{Id}_I\}})$.

}
\end{defn}

\subsection{\normalsize Lifting Free Exponential of $\cC$ to
$\T{J}{\cC}$}  \label{sect1.3}
From now on in this subsection, the category $\cC$
is supposed to satisfy the four conditions of Definition
\ref{MTTfex}. 
The equalisers $\oble{A}{n}$s and the limit $\oble{A}{\infty}$ of $\cC$ are lifted respectively 
to $\oble{\msbf{A}}{n}$s and $\oble{\msbf{A}}{\infty}$ in any orthogonality category $\T{J}{\cC}$ 
using the equaliser and the limit actions
on $\cC$-homset (Propositions \ref{PropequalnT} and \ref{proplimT}).
It is necessary to impose certain conditions on $p_{\infty, n}$
and on $\cC$-morphisms in order to guarantee the 
distributivity of the monoidal product over the limit L$_{\msbf{A}}$
(in Proposition \ref{condist})
as well as that over the equalisers E$_{\msbf{A}}$ (
at \ref{conditiontnteneq}).

\bigskip




\begin{prop}[equaliser $\oble{\msbf{A}}{n}$ for
E$_{\msbf{A}}$ in $\T{J}{\cC}$] \label{PropequalnT}
In $\T{J}{\cC}$ for every object $\msbf{A}=(A, \msbf{A}_p)$, 
the following object $\oble{\msbf{A}}{n}$ with $\eq_A$ becomes the equaliser of
the $n!$-symmetries of $(\msbf{A} \& \msbf{I})^{\otimes n}$:
\begin{align}
& \oble{\msbf{A}}{n} = (\oble{A}{n}, (\oble{\msbf{A}}{n})_p) \quad
 \mbox{with} \nonumber  \\
(\oble{\msbf{A}}{n})_p
& := 
\ccrc{ 
\left\{ \eq \backslash h \mid h \in 
(\gobwt{\msbf{A}}{n}
)_p
\, \,  \mbox{equalises the $n!$-symmetries of 
$\gobwt{\msbf{A}}{n}$} 
\right\} } \label{equalnT}
\end{align}
\end{prop}
\begin{proof}
In the proof, $X_{\bullet} := X \& I$ and 
$X_{\bullet}^{\otimes n}$ is a short for $(X_{\bullet})^{\otimes n}$
either in $\cC$ or
$\T{J}{\cC}$.
By the definition of the morphisms of
the double glueing category, note first: If
a morphism $h$ of co-domain $\msbf{X}^{\otimes n}$
equalises the $n!$-symmetries of $\msbf{X}^{\otimes n}$
in $\T{J}{\cC}$, then $h$ 
does so the $n!$-symmetries of $X^{\otimes n}$ in $\cC$.
The following three conditions need to be checked: \\
\noindent (i) The $\cC$-morphism $\eq_A \backslash h$ resides in $\T{J}{\cC}$
for any $h: \msbf{B} \longrightarrow \msbf{A}_{\bullet}^{\otimes n}$ equalising
the $n!$-symmetries in $\T{J}{\cC}$: For any $b \in \msbf{B}_p$,
$\begin{aligned}
(\eq_A  \backslash h ) \Comp b =
\eq_A \backslash (h \Comp b) 
\end{aligned}$, which belongs to (inside the scope 
	   $\ccrc{}$ of) (\ref{equalnT})  
as $h \Comp b$ belongs to $(\msbf{A}_{\bullet}^{\otimes n})_p$
and equalises the $n!$-symmetries of $A_{\bullet}^{\otimes n}$
by the first note.

\smallskip

\noindent (ii) The $\cC$-morphism $\eq_A$ resides in $\T{J}{\cC}$:
($\bar{\mbox{p}}$) condition holds
directly by the definition (\ref{equalnT}).

\smallskip

\noindent (iii) Any $\T{J}{\cC}$-morphism
$h: \msbf{B} \longrightarrow \msbf{A}_{\bullet}^{\otimes n}$
equalising the $n!$-symmetries factors
	 via $\oble{\msbf{A}}{n}$: \\
By the first note, $h: B \longrightarrow A_{\bullet}^{\otimes n}$
factors via $\oble{A}{n}$ in $\cC$. But by (i), the factorisation
is that for $\T{J}{\cC}$. 
\end{proof}

\begin{rem}[on  Proposition \ref{PropequalnT}] \label{remPropequalnT}
{\em The homset $(\oble{\msbf{A}}{n})_p$
of (\ref{equalnT}) in particular contains the following homset (but not
 vice versa in general) \\
$\begin{aligned}
 \{ 
\xymatrix{I^{\otimes n } \cong  I \ar[r]^{ \epsilon_n}
\ar@/_9pt/[rr]_{ \eq_A \backslash (f \& I)^{\otimes n}}
&  \oble{I}{n} \ar[r]^{\oble{f}{n}} & \oble{A}{n}}
\, \, \mid \, \, I \stackrel{f}{\longrightarrow} A \in \msbf{A}_p
\}, \quad
&
\mbox{where $\epsilon_n$ is 
$\xymatrix{I \cong I^{\otimes n} \ar[rr]^{\eq_I \backslash (I \&
 I)^{\otimes n}} &  & \oble{I}{n}}$.}
 \end{aligned}$
}\end{rem}
\begin{lem} The morphism $p_{n+1,n}$ resides in $\T{J}{\cC}$
so that it is a morphism from $\oble{\msbf{A}}{n+1}$ to
$\oble{\msbf{A}}{n}$.
\end{lem}
\begin{proof}{}
By virtue of the condition ($\overline{\mbox{p}}$) shown to hold
by the definition (\ref{equalnT}).
\end{proof}
By this lemma and Proposition \ref{PropequalnT},
$\{ p_{n+1,n} \}_n$ becomes a diagram for L$_{\msbf{A}}$
in $\T{J}{\cC}$. Then 
\begin{prop}
[limit $\oble{\msbf{A}}{\infty}$ for
L$_{\msbf{A}}$ in $\T{J}{\cC}$]
\label{proplimT} ~\\
In $\T{J}{\cC}$ for every object $\msbf{A}=(A, \msbf{A}_p)$, 
the following $\oble{\msbf{A}}{\infty}$
with $\{ p_{\infty, n}: \oble{\msbf{A}}{\infty} \longrightarrow
 \oble{\msbf{A}}{n} \}_n$ becomes the limit for the sequential diagram
$\{ p_{n+1,n}: \oble{\msbf{A}}{n+1} \longrightarrow \oble{\msbf{A}}{n} \}_n$:
\begin{align}
 & \oble{\msbf{A}}{\infty} := 
\left( \oble{A}{\infty}, 
(\oble{\msbf{A}}{\infty})_p
 \right)
  \quad \mbox{with} \nonumber  \\
&  (\oble{\msbf{A}}{\infty})_p := 
\left\{  \begin{aligned}
 x_{\infty} :  I \longrightarrow \oble{A}{\infty}
\mid \, \{ x_n: \msbf{I} \longrightarrow \oble{\msbf{A}}{n} \}_n \, \, 
\mbox{is a cone to } \\
\mbox{the diagram $\{ p_{n+1,n} \}_n$ in 
$\T{J}{\cC}$}  \}
\end{aligned} \right\} ,  \label{liminf} \\
& \mbox{where $x_\infty$ denotes the mediating $\cC$-morphism for the
 forgetful image of the cone $\{x_n
\}_n$ in $\cC$.}
 \nonumber \end{align}
\end{prop}
See the following diagram how a generator $x_\infty$
belonging to (\ref{liminf})
arises as the limit of
$\{ x_n \}_{n \in \mathbb{N}}$ forgetting in $\cC$.
$$
\xymatrix@R=8pt
@C=20pt
{ \cdots
& \ar[l]
\oble{A}{n} &  \ar[l]_{p_{n+1,n}} \oble{A}{n+1}
 \cdots       &  \cdots  & \oble{A}{\infty}
\ar@/_20pt/[lll]_(.5){p_{\infty, n}} \\
&  &  
      &   &  I \ar[u]_{x_\infty} \ar@/^9pt/[ulll]_{x_n \in 
(\oble{\msbf{A}}{n})_p}
}
$$
The diagram describes the arrow $x_n \in 
(\oble{\msbf{A}}{n})_p$ because
$x_n: \msbf{I} \longrightarrow \oble{\msbf{A}}{n}$.

\begin{proof}
First, the following two conditions need to be checked:

\noindent (i)
$p_{\infty, n}$ resides in $\T{J}{\cC}$: Direct by the definition (\ref{liminf})
.

\noindent (ii)
Any mediating morphism $\tau_\infty$ in $\cC$ resides in $\T{J}{\cC}$: \\
Let $\{ \tau_n : \msbf{C} \longrightarrow
\oble{\msbf{A}}{n} \}_n$ be any cone to
the diagram $\{p_{n+1,n}  \}_n$ in $\T{J}{\cC}$. Then $\cC$ has the mediating
	$\tau_\infty: C \longrightarrow \oble{A}{\infty}$ for 
the forgetful image of the cone in $\cC$.
Then $\tau_\infty \Comp c \in (\oble{A}{\infty})_p$ needs to be shown
for any $c \in  \msbf{C}_p$. For this, it suffices to show (ii-i)
$\{ \tau_n \Comp c: \msbf{I} \longrightarrow \oble{\msbf{A}}{n} \}$ is a cone to the diagram $\{ p_{n, n+1} \}$ in $\T{J}{\cC}$, and 
(ii-ii) $\tau_\infty \Comp c$ is the $\cC$-mediating to the forgetful image of the cone. (ii-i) is direct as $c: \msbf{I} \longrightarrow \msbf{C}$
and (ii-ii) holds as $\tau_n \Comp c = (p_{\infty, n} \Comp \tau_\infty) \Comp c =
p_{\infty, n} \Comp (\tau_\infty \Comp c)$.

Second, (\ref{liminf}) is shown to be closed under the double orthogonal.
For this, observe 
$p_{\infty, n} \Comp (\ref{liminf}) \subset \oble{\msbf{A}}{n}_p,$
which implies by the reciprocity
$\crc{(\ref{liminf})} \supset \crc{(\oble{\msbf{A}}{n}_p)} \Comp p_{\infty,n}$
for all $n$.
This means $\ort{z}{}{\crc{(\ref{liminf})}}$ implies
$\ort{z}{}{\crc{(\oble{\msbf{A}}{n}_p)} \Comp p_{\infty,n}}$,
then by reciprocity 
$\ort{p_{\infty,n} \Comp z}{}{\crc{((\oble{\msbf{A}}{n})_p)}}$,
thus $p_{\infty,n} \Comp z \in \ccrc{((\oble{\msbf{A}}{n})_p)}=
(\oble{\msbf{A}}{n})_p$ for all $n$. This concludes $z$,
as the mediating
for the cone $\{ p_{\infty,n} \Comp z \}_n$, 
belongs to (\ref{liminf}).
\end{proof}

\begin{rem}[on Proposition \ref{proplimT}] \label{remproplimT}
{\em  The homset $(\oble{\msbf{A}}{\infty})_p$
of (\ref{liminf}) in particular contains the following homset
(but not vice versa in general)
\begin{align*}
 &
\{  
\xymatrix{I  \ar[r]^{\epsilon_\infty}
& \oble{I}{\infty}
 \ar[r]^{\oble{f}{\infty}} 
 & \oble{A}{\infty} }
\, \, \mid \, \, I \stackrel{f}{\longrightarrow} A \in \msbf{A}_p
\},
\end{align*}
in which $\epsilon_\infty$ is the universal $\cC$-morphism for the cone
$\{ 
\xymatrix{I \cong I^{\otimes n} \ar[r]^{\epsilon_n}
&   \oble{I}{n}} \}_{n \in \mathbb{N}} $
on the limit L$_I$. \\
The remark holds because
$\oble{f}{\infty} \Comp \,  \oble{\epsilon}{\infty}$
is the universal morphism of the cone
$\{  
\xymatrix{I  \ar[r]^{\epsilon_n}
& \oble{I}{n}
 \ar[r]^{\oble{f}{n}} 
 & \oble{A}{n} }
\, \, \}_n$,
whose each member belongs to $(\oble{\msbf{A}}{n})_p$
by Remark \ref{remPropequalnT}.
}\end{rem}

\bigskip

In $\T{J}{\cC}$, neither distributivity of
the monoidal product over the equaliser $\oble{\msbf{A}}{n}$
nor over the limit $\oble{\msbf{A}}{\infty}$ are retained in general.
Hence we need to augment the following respective conditions
in terms of the orthogonality and the monoidal product:

Recall the distribution of $\otimes$ 
over E$_{\msbf{A}}$ in $\T{J}{\cC}$
stipulates the following two (i) $\eq_A \otimes B$ lives in $\T{J}{\cC}$,
and (ii) $(\eq_A \otimes B) \backslash g$
lives in $\T{J}{\cC}$ for any $g$ with the codomain 
$\gobwt{\msbf{A}}{n} \otimes \msbf{B}$
equalising the ($n!$-symmetries)$\otimes \msbf{B}$ in $\T{J}{\cC}$. 
The stipulation (i) is tautological as the morphism is checked automatically
to satisfy ($\overline{\mbox{p}}$) condition. \\
On the other hand, (ii) stipulates 
\begin{align*}
\{ v \in (\gobwt{\msbf{A}}{n} \otimes \msbf{B})_p  
\mid
\mbox{
$v$ equalises
 (the $n!$-symmetries) $\otimes \msbf{B}$ }  \} \subset  
(\eq_A \otimes B) \Comp 
(\oble{\msbf{A}}{n}
\otimes \msbf{B})_p   
\end{align*}
This is equivalent to the following by $v = \eq \otimes B 
\Comp ((\eq \otimes B)  \backslash v)$
and by the reciprocity on the premise: 
\begin{align*}
\ort{(\eq \otimes B)  \backslash v}{}{
\crc{(
(\gobwt{\msbf{A}}{n})_p \otimes \msbf{B}_p
)}
\Comp (\eq \otimes B)
}
\Longrightarrow
\ort{(\eq \otimes B)  \backslash v}{}{
\crc{(
(\oble{\msbf{A}}{n})_p \otimes \msbf{B}_p
)}}
\end{align*} 
This is rewritten equivalently to
\begin{align}
\crc{(\crceq{(
(\gobwt{\msbf{A}}{n})_p \otimes \msbf{B}_p
)}
\Comp (\eq \otimes B)
)}
\subset  
\ccrc{(
(\oble{\msbf{A}}{n})_p \otimes \msbf{B}_p
)}, \label{conditiontnteneq} 
\end{align}
in which $\crceq{W}$ is the subset of $\crc{W}$
consisting of elements equalising 
the (the $n!$-symmetries) $\otimes \msbf{B}$.
Note the equalisation for $\crc{W}$ is via precomposing
the symmetries $\otimes B$, reciprocally for $v$ via composing
them.

To conclude, (\ref{conditiontnteneq}) is the condition 
for $\otimes$ to distribute over 
E$_{\msbf{A}}$. 
\smallskip

\begin{exam}[of the condition (\ref{conditiontnteneq})]
\label{excondiT}
If $\eq_A$ has left inverse $\eq_A^{\flat}$ in $\cC$ so that  
$\eq_A^{\flat} \Comp\eq_A = \oble{A}{n}$, then the following condition
(\ref{strngconditiontnteneq}) is satisfied,
which is stronger than (\ref{conditiontnteneq}).
\begin{align}
\crc{
(\eq_A \Comp (\oble{\msbf{A}}{n})_p
\otimes \msbf{B}_p)
} \Comp 
(\eq \otimes B)
\supset
\crc{((\oble{\msbf{A}}{n})_p \otimes \msbf{B}_p)}
\label{strngconditiontnteneq} 
\end{align}
\end{exam}
\begin{proof}{}

\noindent (Claim 1):
(\ref{strngconditiontnteneq}) implies (\ref{conditiontnteneq})
in general. 

\smallskip
To prove the claim, the following two subclaims are shown:

\noindent (subclaim 1.1):
For any $U \subset (A \& I)^{\otimes n}$
and $V \subset B$,
$$\crceq{(U \otimes V)}=\crc{(U^\eq \otimes V)},$$
where $U^\eq$ is the subset of $U$  
consisting of elements equalising 
the $n!$-symmetries.

\noindent(proof of 1.1)
By the reciprocity of the orthogonality
$\ort{
(\sigma \otimes B)(g)}{}{U \otimes V} \, \mbox{iff} \, 
\ort{g}{}{
\sigma(U) \otimes V}$ for any $\sigma \in \mathfrak{S}_n$,
for which the left (resp. right) $\sigma$ 
acts for composing (resp. for precomposing).

\noindent (subclaim 1.2):
\begin{align}
\crceq{
(\gobwt{\msbf{A}}{n})_p \otimes \msbf{B}_p
)
}=\crc{((
(
\msbf{A} \& \msbf{I})_p)^{\otimes_\eq n} \otimes \msbf{B}_p)
}, \label{eqclaim2}
\end{align}
where $V^{\otimes_\eq n}$ is a short for $(V^{\otimes n})^\eq$. 

\noindent(proof of 1.2)
By subclaim 1.1 applied to the following stable tensor
as $\msbf{B}_p=\ccrc{\msbf{B}_p}$: \\ 
$\crc{(\gobwt{\msbf{A}}{n})_p \otimes \msbf{B}_p)}
=
\crc{(
\ccrc{(((\msbf{A} \& \msbf{I})_p
)^{\otimes n})}
\otimes \msbf{B}_p
)}
=
\crc{(((\msbf{A} \& \msbf{I})_p
)^{\otimes n}  \otimes \msbf{B}_p)}
$.

\smallskip

\noindent(proof of Claim 1) 
\begin{align*}
\eq_A \Comp (\oble{\msbf{A}}{n})_p 
& \supset
((\msbf{A} \& \msbf{I})_p)^{\otimes_\eq n} \tag*{by (\ref{equalnT})} \\
\crc{(\eq_A \Comp (\oble{\msbf{A}}{n})_p \otimes \msbf{B}_p)}
&  \subset
\crc{
(((\msbf{A} \& \msbf{I})_p)
^{\otimes_\eq n}  \otimes \msbf{B}_p)
} \tag*{by orhogonality} \\
\crc{(
(\oble{\msbf{A}}{n})_p \otimes \msbf{B}_p
)}
& \subset 
\crc{
(((\msbf{A} \& \msbf{I})_p)
^{\otimes_\eq n}  \otimes \msbf{B}_p)
}
\Comp (\eq \otimes B)
\tag*{precomposed with $\eq \otimes B$, then by (\ref{strngconditiontnteneq})}\\
\crc{(
(\oble{\msbf{A}}{n})_p \otimes \msbf{B}_p
)}
& \subset
\crceq{(
(\gobwt{\msbf{A}}{n})_p \otimes \msbf{B}_p
)}
\Comp (\eq \otimes B) \tag*{by (\ref{eqclaim2})},
\end{align*}
which implies (\ref{conditiontnteneq}) by the orthogonality.

\bigskip

\noindent (Claim 2): 
If $\eq_A$ has left inverse $\eq_A^{\flat}$ in $\cC$, then 
(\ref{strngconditiontnteneq})  holds.\\
In the proof 
$\blacklozenge$ abbreviates 
$\eq_A \Comp (\oble{\msbf{A}}{n})_p$.
Any $y$ from RHS is written by 
$y=y' \Comp (\eq_A \otimes B)$ with $y'=y \Comp (\eq_A^{\flat} \otimes B)$
in terms of the left inverse.
We need to prove the following first orthogonality:
$\ort{y'}{}{
\blacklozenge 
\otimes \msbf{B}_p} \quad \mbox{iff} \quad 
 \ort{y}{}{
(\eq_A^{\flat} \otimes B)
\Comp 
( \blacklozenge 
\otimes \msbf{B}_p} )=
(\eq_A^{\flat} \Comp  \blacklozenge 
)
\otimes 
\msbf{B}_p$,
whose second orthogonality holds by the choice $y$ as 
$\eq_A^{\flat} \Comp 
\blacklozenge 
=(\oble{\msbf{A}}{n})_p$.
\end{proof}

\bigskip
On the other hand, for the distributivity over the limit,
\begin{prop}[condition for monoidal product $\otimes$ to distribute over the
  limit in $\T{J}{\cC}$] \label{condist}
The following condition 
in $\T{J}{\cC}$ 
is necessary and sufficient for the distribution of
  $\otimes$ over the limit L$_{\msbf{A}}$. 
\begin{align}
\ccrc{((\oble{\msbf{A}}{\infty})_p \otimes \msbf{B}_p)
}
=
\bigcap_{n \in \msbf{N}}
(p_{\infty, n} \otimes B)^{-1} \Comp 
\ccrc{((\oble{\msbf{A}}{n})_p \otimes \msbf{B}_p)
} \label{condition}
\end{align}
Notation: $f^{-1} \Comp U : = 
\{ x : I \longrightarrow X \mid f \Comp x \in U  \}$
for a morphism $f: X \longrightarrow Y$ in $\cC$
and a homset $U \subseteq \cC(I, Y)$.

The condition when
$\msbf{B}=\msbf{I}$ is automatically valid for any
$\T{J}{\cC}$. But not necessarily so for a general object $\msbf{B}$.
\end{prop} 
\begin{proof}
First note that  $\subset$ of (\ref{condition}) is tautological for any $\cC$
by definitions (\ref{equalnT}) and (\ref{liminf}).
Hence the condition (\ref{condition}) 
is equivalent to the following (\ref{altcondition}):
\begin{align}
\forall n \in \msbf{N} \, \, 
(p_{\infty, n} \otimes B) \Comp u \in
\ccrc{((\oble{\msbf{A}}{n})_p \otimes \msbf{B}_p)} \Longrightarrow 
u 
\in \ccrc{((\oble{\msbf{A}}{\infty})_p \otimes \msbf{B}_p)}
\label{altcondition}
 \end{align}
The condition (\ref{altcondition}) says that $x_\infty$
becomes the mediating morphism
for the cone $\{ x_n : \msbf{I} \longrightarrow \oble{\msbf{A}}{n}
\otimes \msbf{B} \}$ in $\T{J}{\cC}$,
hence derives the necessity and sufficiency.
In particular, (\ref{altcondition}) when $\msbf{B}=\msbf{I}$
is Proposition \ref{proplimT}, hence is valid in any $\T{J}{\cC}$.
\end{proof}

The main theorem of this section is obtained 
by Proposition \ref{condist}.
\begin{thm}[Free Exponential in $\T{J}{\cC}$] \label{mainsec1}
Suppose an orthogonality on a monoidal category $\cC$
is focused and satisfies the conditions
(\ref{conditiontnteneq}) and 
(\ref{condition}). Then,
whenever $\cC$ has the free exponential constructed by Definition
\ref{MTTfex}, it is also true for
the orthogonality category $\T{J}{\cC}$ so that
forgetful $\T{J}{\cC} \longrightarrow \cC$ preserves 
the free exponentials. 
\end{thm}

\section{\large Monoidal Category $\TsKer$ 
of s-finite Transition Kernels with Biproducts} \label{sect2}
This section concerns a measure theoretic study,
independent from  Section \ref{sect1}. 
The main sources of the section are Staton \cite{Stat}
and Hamano \cite{HamLEC}.  We also refer to Bauer's book \cite{Bau}
for general measure theory. 
Preliminary notions from measure theory
are also referenced in Panangaden's book  \cite{Prabook}.

\subsection{\normalsize Preliminaries from Measure Theory} \label{PfMT}
This subsection recalls some basic definitions and the monotone convergence theorem
from measure theory, necessary in this paper.

\noindent ({\bf Terminology})
$\mathbb{N}$ denotes the set of non negative integers.
$\mathbb{R}_+$ denotes the set of non negative reals.
$\zeroinf$ denotes $\mathbb{R}_+ \cup \{ \infty \}$.
$\mathfrak{S}_n$ denotes the symmetric group
over $\{1, \ldots ,n\}$.
For a subset $A$, $\chi_A$ denotes the characteristic function of $A$.
$\delta_{x,y}$ is the Kronecker delta.
$\uplus$ denotes the disjoint union of sets.

\smallskip

\begin{defn}[$\sigma$-field $\mc{X}$
and measurable space $\ms{X}$]~\\{\em
A {\em $\sigma$-field} over a set $X$
is a family $\mc{X}$ of subsets of $X$
containing $\emptyset$,
closed under the complement and countable union.
A pair $\ms{X}$ is called a {\em measurable space}.
The members of $\mc{X}$ are called {\em measurable sets}.
The measurable space is often written simply 
by $\mc{X}$, as $X$ is the largest element in $\mc{X}$.
For a measurable set $Y \in  \mc{X}$,
the measurable subspace $\mc{X} \cap Y$,
called the {\em  restriction on} $Y$,
is defined by $\mc{X} \cap Y := \{  A \cap Y \mid A \in \mc{X} \}$.
}
\end{defn}
\smallskip

\begin{defn}[$\sigma(\mc{F})$ and Borel $\sigma$-field $\borel$]
{\em For a family $\mc{F}$ of subsets of $X$,
$\sigma(\mc{F})$ denotes the $\sigma$-field generated by $\mc{F}$,
i.e., the smallest $\sigma$-field containing $\mc{F}$.
When $X$ is $\zeroinf$ and $\mc{F}$ is the family $\mc{O}_{\zeroinf}$
of the open sets in $\zeroinf$ (with the topology whose basis consists of 
the open intervals in $\mathbb{R}_+$
together with $(a, \infty):=\{ x \mid a  < x  \}$
for all $a \in \mathbb{R}_+$), the $\sigma$-field
is denoted by $\borel$, whose members
are called Borel sets over $\zeroinf$.
}\end{defn}
\smallskip
\begin{defn}[measurable function]{\em
For measurable spaces$\ms{X}$ and $\ms{Y}$,
a function $f: X \longrightarrow Y $ is {\em
 $(\mc{X},\mc{Y})$-measurable}
(often just {\em measurable})
if $f^{-1}(B) \in \mc{X}$ whenever $B \in \mc{Y}$.
In this paper,
a measurable function unless otherwise mentioned is  to the Borel set
$\borel$ over
$\zeroinf$ from some measurable
space $\ms{X}$. 
}\end{defn}
\smallskip
\begin{defn}[measure]{\em
A {\em measure} $\mu$ on a measurable space $\ms{X}$
is a function from $\mc{X}$ to $\zeroinf$
satisfying ({\em $\sigma$-additivity}):
If $\{ A_i \in \mc{X} \mid i \in I \}$
is a countable family of pairwise disjoint sets,
then $\mu(\bigcup_{i \in I} A_i)= \sum_{i \in I} \mu(A_i)$.
 }\end{defn}
\smallskip

\begin{notn}[Lebesgue integration (cf. Chapter 3.1 of \cite{Prabook})]{\em
For a measure $\mu$ on $\ms{X}$, and a
$(\mc{X}, \borel)$ -measurable function $f$, 
the Lebesgue integral of $f$  over $X$ wrt the measure $\mu$ is denoted by
$\int_X f(x) \mu(dx)$, which is simply written $\int_X f d \mu$.
It is also written $\int_X d \mu f$.
}\end{notn}

\begin{thm}[monotone convergence] \label{MC}
Let $\mu$ be a measure on a measurable space $\ms{X}$.
For an monotonic sequence $\{ f_n \}$
of $(\mc{X}, \borel)$-measurable functions,
if $f = \sup_n {f_n}$, then $f$ is measurable and $\sup \int_X f_n d \mu =
\int_X f d \mu$.
\end{thm}

\smallskip

\begin{defn}[push forward measure $\mu \Comp F^{-1}$
along a measurable function $F$]\label{pfm}{\em
For a measure $\mu$ on $\ms{Y}$ and a measurable function $F$ from
$\ms{Y}$ to $\ms{Y'}$, $\mu'(B'):= \mu (F^{-1}(B'))$  with $B' \in \mc{Y'}$
becomes a measure on
$\ms{Y'}$, called {\em push forward measure} of $\mu$
along $F$. The push forward measure $\mu'$ has the following property
for any measurable function $g$ on $\ms{Y'}$, called
{\em ``variable change of integral along push forward $F$''}:
\begin{align} 
& \textstyle \int_{Y'} g \, d \mu'= \int_{Y} (g \Comp F) \, d \mu. 
& \text{That is,} \quad \quad
\textstyle \int_{Y'} g (y') \,\mu' (d y')= \int_{Y} g (F(y)) \, \mu (dy)
\label{PFM}
\end{align}
 The push forward measure $\mu'$ is often denoted by $\mu \Comp F^{-1}$
 by abuse of notation. \\
Note: The abuse of notation will be resolved 
category theoretically in 
Section \ref{sect.4.1.1}.}
\end{defn}

\subsection{\normalsize Transition Kernels}
\begin{defn}[transition kernel \cite{Bau}]\label{TKer}{\em 
For measurable spaces $(X, \mathcal{X})$ and $(Y, \mathcal{Y})$,
a {\em transition kernel} from
$(X, \mathcal{X})$ to $(Y, \mathcal{Y})$ is a function
$$\begin{aligned}
\kappa  :  X \times \mathcal{Y}
\longrightarrow \zeroinf \quad \mbox{satisfying}
\end{aligned}$$ 
\begin{itemize}
 \item[(i)]
For each $x \in X$, 
the function $\kappa(x, -):
\mathcal{Y}
\longrightarrow \zeroinf$ is a measure on $\ms{Y}$.
\item[(ii)]
For each $B \in \mathcal{Y}$,
 the function $\kappa(-,B): X
\longrightarrow \zeroinf$ is measurable on $\ms{X}$.
\end{itemize}
}\end{defn}

\begin{defn}[operations $\kappa_*$ and $\kappa^*$ of a kernel $\kappa$
on measures and measurable functions] \label{opekernel} ~\\
{\em Let $\kappa: \ms{X} \longrightarrow \ms{Y}$
be a transition kernel.
\begin{itemize}
 \item 
For a measure $\mu$ on $\mathcal{X}$, 
$$\begin{aligned}
(\kappa_* \mu)(B) := \int_{X} \kappa (x,B) \mu( dx)
\end{aligned}$$
is a measure on $\mathcal{Y}$, where $B \in \mathcal{Y}$.
\item
For a measurable function $f$ on $\mathcal{Y}$, 
$$\begin{aligned}
(\kappa^* f)(x) := \int_{Y} f(y) \kappa (x, dy)
\end{aligned}$$
is measurable on $\mathcal{Y}$, where $x \in X$.
In particular, for a characteristic function $\chi_B$
for any $B \in \mathcal{Y}$,
\begin{align*}
(\kappa^* \chi_B)(x) := \kappa (x, B) 
\end{align*}
\end{itemize}
}\end{defn}
\noindent It is direct to check, by the monotone
convergence theorem \ref{MC}, that $\kappa^* f$ is measurable.

\begin{defn}[category $\TKer$ of transition kernels] ~
{\em 
$\TKer$ denotes the category where each object is a measurable space $\ms{X}$
and a morphism is a transition kernel $\kappa(x, B)$
from $\ms{X}$ to $\ms{Y}$.
The composition is the 
{\em convolution} of two kernels
$\kappa(x, B) : \ms{X} \longrightarrow
\ms{Y}$ and $\iota(y, C) : \ms{Y} \longrightarrow
\ms{Z}$:
\begin{align}
\iota \Comp \kappa (x,C) = &
\int_Y \kappa(x,dy) \iota(y,C) \label{compTKer}
\end{align}
$\operatorname{Id}_{\ms{X}}$ is {\em Dirac delta measure} $\delta:
\ms{X} \longrightarrow \ms{X}$, defined by
for $x \in X$ and $A \in \mc{X}$;
$$\begin{aligned}
\text{if $x \in A$ then $\delta(x,A) = 1$, else $\delta(x,A) = 0$.}
\end{aligned}$$  
}\end{defn}

\begin{rem}\label{remmesfun}{\em
Measures and measurable functions
both reside as morphisms in $\TKer$: Let $\ms{I}$ be the singleton
 measurable space with $I=\{ * \}$, hence $\mc{I}=\{ \emptyset, \{ *\}
 \}$, then 
$$\begin{aligned}
&  \TKer(\mc{I}, \mc{X}) &   = \{ \lambda A \in \mc{X}. \, \kappa(*,A) \mid 
\mbox{$\kappa$ is a kernel from $\mc{I}$} \} 
 =  \{ \mbox{ the measures $\mu$ on $\ms{X}$ }\} \\
&  \TKer(\mc{X}, \mc{I}) &  =  \{ \lambda x \in X . \, \kappa(x,\{ * \}) \mid
\mbox{$\kappa$ is a kernel to $\mc{I}$} \} \cup
   \{ \lambda x \in X . \, \kappa(x, \emptyset) =0 :
X \rightarrow \zeroinf \} \\
& & = \{ \mbox{\rm the measurable functions $f$ on $\ms{X}$ to $\borel \!$ }\} 
 \end{aligned}$$
The operations $\kappa_*$ and $\kappa^*$ of Definition 
\ref{opekernel} are respectively categorical precomposition and composition
with $\kappa$ in $\TKer$ so that $\begin{aligned}
\kappa_* \mu = \kappa \Comp \mu
\quad \text{and} \quad \kappa^* f = f \Comp \kappa.
\end{aligned}$
}\end{rem}

In the sequel, when a transition kernel $\kappa$ has the domain
(resp. co-domain) $\mc{I}$ in $\opTsKer$, then $\kappa(\{ * \}, x)$ (resp. $\kappa(X, *)$)
is simply written as a measurable function $\kappa(x)$ (resp.
measure $\kappa(X)$).

\begin{rem}[$\SRel$ \cite{Prabook,Prapaper}]
{\em The category $\SRel$ of {\em stochastic relations}
is a wide subcategory of $\TKer$ strengthening
the conditions of Definition \ref{TKer}
into (i) $\kappa(x, -)$ is a {\em (sub)probability} measure
(i.e., its domain is $[0,1]$)
 and (ii) $\kappa(-, B)$ is {\em bounded} measurable. 
The morphisms in $\SRel$ are called {\em (sub)Markov kernels}.
}\end{rem}
\noindent It is now well known that the composition (\ref{compTKer})
for Markov kernels comes from Giry's probabilistic monad,
generalising the power set monad of the relational composition
\cite{Giry, Prabook}.

\subsection{\normalsize Countable Biproducts in $\TKer$}

\begin{prop}[biproduct $\coprod$] \label{TKerbipro}
$\TKer$ has countable biproducts $\coprod$, which are 
defined for a countable family $ \{  \msi{X}{i} \}_i$ of measurable spaces
as follows:
\begin{align*}
&  \textstyle \coprod \limits_{i} \msi{X}{i} :=(\bigcup \limits_i \{ i
 \} \! \times \! X_i, \,  \biguplus \limits_i
\mathcal{X}_i),  
\end{align*}
where $
\biguplus_i \mathcal{X}_i := \{ \bigcup_i \{ i \} \! \times \! A_i \mid A_i \in
\mathcal{X}_i  \}$
is the $\sigma$-field generated by the measurable sets of each
summands.
\end{prop}
\noindent Consult the proof of Proposition 2.9 of \cite{HamLEC}
for the same assertion.

\smallskip

\noindent The unit of the biproduct is
the null measurable space 
$\mc{T}= (\emptyset, \{ \emptyset \})$.

\subsection{\normalsize Monoidal Product and Countable Biproducts in $\TsKer$ }


\begin{defn}[product of measurable spaces \cite{Bau, Prabook}] \label{prodobj}{\em
The product of measurable spaces $(X_1, \mc{X}_1)$ and $(X_2, \mc{X}_2)$
is the measurable space $(X_1 \times X_2, \mc{X}_1 \otimes \mc{X}_2)$,
where $\mc{X}_1 \otimes \mc{X}_2$ denotes the $\sigma$-field
over the cartesian product $X_1 \times X_2$
generated by {\em measurable rectangles} $A_1 \times A_2$'s 
such that $A_i \in \mc{X}_i$.}
 \end{defn}
In order to accommodate 
measures into the product of measurable spaces,
each measure $\mu_i$ on $(X_i, \mc{X}_i)$ needs to be extended uniquely
to the product.
The condition of {\em $\sigma$-finiteness} ensures this,
yielding the unique product measure over the product measurable space:
\begin{defn}[$\sigma$-finiteness \cite{Bau, Prabook}]{\em
A measure $\mu$ on $\ms{X}$
is {\em $\sigma$-finite} when the set $X$ is written as a countable union of sets of
finite measures. That is, $\exists A_1, A_2, \ldots \in \mc{X}$ such that
$\mu(A_i) < \infty$ and $X = \cup_{i=1}^{\infty} A_i$.
}\end{defn} 

\begin{defn}[product measure \cite{Bau, Prabook}]\label{prodmeasure}{\em
For $\sigma$-finite measures $\mu_i$ on $(X_i, \mc{X}_i)$ with $i=1,2$,
there exists a unique
measure $\mu$ on $(X_1 \times X_2, \mc{X}_1 \otimes \mc{X}_2)$
such that $\mu (A_1 \times A_2)=\mu_1 (A_1) \mu_2 (A_2)$.
$\mu$ is written $\mu_1 \otimes \mu_2$ and called the {\em product measure}
of $\mu_1$ and $\mu_2$.}\end{defn}

The product measure derived from $\sigma$-finite measures
guarantees the basic theorem in measure theory, stating
double integration is treated as iterated integrations.
\begin{thm}[Fubini-Tonelli \cite{Bau, Prabook}] \label{FT}
For $\sigma$-finite measures $\mu_i$ on $(X_i, \mc{X}_i)$ with $i=1,2$
and a 
$(\mc{X}_1 \otimes \mc{X}_2, \borel)$-measurable function $f$,
$$\begin{aligned}
\int_{X_1 \times X_2} \! \! \! \! \! \! f \, d (\mu_1 \otimes \mu_2) =
\int_{X_2} \! \!   d \mu_2  \int_{X_1} \!  f \, d \mu_1 =
\int_{X_1} \! \! d \mu_1  \int_{X_2} \! f \, d \mu_2
\end{aligned}$$
\end{thm}
The Fubini-Tonelli Theorem becomes crucial
category theoretically to us for the following two 
(i) dealing with functoriality of morphisms
on the product measurable spaces (cf. Proposition
\ref{FTsfin} below)
(ii) giving a new instance of the orthogonality
using measure theory (cf. Proposition \ref{adjn}
of Section \ref{sect4}).

\smallskip
Since the $\sigma$-finiteness retained in the category $\TKer$
is not closed under the categorical composition,
we define the two classes of transition kernels,
{\em finiteness} and {\em s-finiteness}, 
respectively by tightening and loosening the $\sigma$-finiteness
so that the both classes are preserved under the composition of $\TKer$.

\begin{defn}[s-finite kernels \cite{Sharpe, Stat}] \label{sfinker}
{\em  
Let $\kappa$ be a transition kernel from $\ms{X}$ to $\ms{Y}$.
\begin{itemize}
 \item[-]
$\kappa$ 
is called {\em finite} when 
$
\sup_{x \in X} \kappa(x, Y) < \infty
$; i.e., the condition says that up to the scalar $0 < a < \infty$ factor
determined by the sup, $\kappa$ is Markovian.
\item[-]
$\kappa$ 
is called {\em s-finite} (i.e., sum of finite) when $\kappa = \sum_{i \in \mathbb{N}} \kappa_i$
where each $\kappa_i$ is a finite kernel from $\ms{X}$ to $\ms{Y}$
and the sum is defined by 
$(\sum_{i \in \mathbb{N}} \kappa_i )(x, B)
:=\sum_{i \in \mathbb{N}} \kappa_i (x, B)$.
This is well-defined because 
any countable sum of kernels from $\ms{X}$ to $\ms{Y}$
becomes a kernel of the same type.
\end{itemize}
}\end{defn}
\noindent In the definition of s-finiteness, 
note that
$(\sum_{i \in \mathbb{N}} \kappa_i)^*
= 
\sum_{i \in \mathbb{N}} \kappa_i^*$ and 
$(\sum_{i \in \mathbb{N}} \kappa_i)_*
= 
\sum_{i \in \mathbb{N}} (\kappa_i)_*$
for the operations of Definition \ref{opekernel}:
That is, the preservation of the operation $(~)^*$ (resp. of $(~)_*$) 
means the commutativity of integral
over countable sum of measures (resp. of measurable functions).
\begin{rem} \label{clopf}
{\em The both classes of the finite kernels and of the s-finite kernels are
closed under the categorical composition of $\TKer$.
This is directly calculated for the finite kernels, to which the
s-finite ones are reduced by virtue of
the note in the above paragraph. We refer 
to the proof of Lemma 3 of \cite{Stat} for the calculation.
}\end{rem}
The original  Fubini-Tonelli (Theorem \ref{FT})
for the $\sigma$-finite measures 
extends to
the s-finite measures:
\begin{prop}[Fubini-Tonelli extending for s-finite measures (cf. 
Proposition 5 of Staton \cite{Stat})]  \label{FTsfin}
Theorem \ref{FT} extends for s-finite measures
$\mu_1$ and $\mu_2$ (with the same $f$).
\end{prop}
\noindent See the proof of Proposition 2.18 for the same proposition.

It is derived
from Proposition \ref{FTsfin} that the s-finite transition kernels form a
monoidal category.
\begin{defn}[monoidal subcategories $\TsKer$ of s-finite kernels] \label{moncat}{\em
$\TsKer$ is a wide subcategory of $\TKer$,
whose morphisms are the {\em s-finite} transition kernels.
$\TsKer$ has a symmetric monoidal product $\otimes$: On objects 
is by Definition \ref{prodmeasure}.
Given morphisms $\kappa_1: \ms{X_{\mbox{\scriptsize $1$}}} \longrightarrow 
\ms{Y_{\mbox{\scriptsize $1$}}}$
and
$\kappa_2: \ms{X_{\mbox{\scriptsize $2$}}} \longrightarrow 
\ms{Y_{\mbox{\scriptsize $2$}}}$,
their product is defined explicitly: 
\begin{align*}
(\kappa_1 \otimes \kappa_2) ((x_1,x_2), C) :=
\! \int_{Y_1}  \! \!  \kappa_1 (x, dy_1) 
\! \int_{Y_2}  \! \!  \kappa_2 (x, dy_2) 
\, \chi_{C} ((y_1,y_2))
\end{align*}
Alternatively, thanks to Fubini-Tonelli (Proposition \ref{FTsfin}),
the product is implicitly defined as 
the unique transition kernel
$\kappa_2 \otimes \kappa_2: (X_1 \times X_1, \mc{X}_1 \otimes \mc{X}_2)
\longrightarrow (Y_1 \times Y_2, \mc{Y}_1 \otimes \mc{Y}_2)$
satisfying the following for any rectangle $B_1 \times B_2$
with $B_i \in \mc{X}_i$:  
\begin{align*}
(\kappa_1 \otimes \kappa_2) ((x_1,x_2), B_1 \times B_2)=
\kappa_1 (x_1,B_1) \kappa_2 (x_2,B_2) \quad
\end{align*}
}\end{defn}
\noindent The unit of the monoidal product is the singleton measurable
space $\ms{I}$.

\begin{prop}[The biproducts] \label{retbip}~\\
$\TsKer$ has countable biproducts
which are those in $\TKer$ residing inside the subcategory.
\end{prop}
\noindent Consult the proof of Proposition 2.20 of \cite{HamLEC}
for the same assertion.

\smallskip

\noindent {\bf Note:} In the sequel, only the product structure
 of $\coprod$ is employed.
  That is Sections \ref{sect3} and \ref{sect4}
 do not treat $\prod$ as biproduct 
but simply as product.
 Accordingly, $\prod$ is written by $\&$.

\section{\large The Free Exponential in $\opTsKer$ } \label{sect3}
This section constructs the free exponential structure
of $\opTsKer$.
The opposite setting is chosen accordingly to
\cite{HamLEC} by virtue of the asymmetry between the first
and the second arguments of the transition kernels.

\noindent {\bf Notation for morphisms in the opposite setting}:
In the opposite $\opTsKer$, a morphism
$\kappa : \ms{X} \longrightarrow
\ms{Y}$ is a transition kernel from $\ms{Y}$ to $\ms{X}$.
Accordingly a morphism $\kappa$ is denoted by $\kappa(A, y)$ meaning that
its left (resp. right) argument determines a measure (resp. a
measurable function). In particular, the Dirac delta
measure which is the identity morphism on $\ms{X}$
is written by $\Dd{x}{A}$.  
Recall  Remark \ref{remmesfun} 
for the opposite category that $\opTsKer(\mc{I}, \mc{X})$ (resp. $\opTsKer(\mc{X}, \mc{I})$)
consists of the measurable functions (resp. the s-finite measures)
on $\ms{X}$. 
The composition of two morphisms 
$\kappa(A, y) : \ms{X} \longrightarrow
\ms{Y}$ and $\iota(B, z) : \ms{Y} \longrightarrow
\ms{Z}$ in $\opTsKer$ is
$\begin{aligned}
\iota \Comp \kappa (A,z) = &
\textstyle \Int{Y}{\kappa(A,y)}{\iota(dy,z)} 
\end{aligned}$.
In what follows, the morphisms of the opposite category are also called kernels.

\smallskip

In our measure-theoretic framework,
the equaliser for the exponential is defined 
slightly more generally,
not only for a rooted object $\mc{X}_\bullet = \mc{X} \& I$ with the monoidal unit $\mc{I}$,
but also for a general object $\mc{X}$ of $\opTsKer$:
\begin{defn}[measurable space 
$\msn{X}{n}$] \label{ms(n)}
{\em
For any measurable space $\ms{X}$
and any natural number $n$,
a measurable space $\msn{X}{n}$
is defined by
\begin{align*}
&  X^{(n)}:= \{ x_1 \cdots x_n \mid x_i \in X \} \quad \mbox{and} \quad  
  \mc{X}^{(n)}  := \{ A \subseteq X^{(n)} \mid \FI (A) \in
	   \mc{X}^{\otimes n}  \}, \\
& \mbox{where $F:
X^{\otimes n} \longrightarrow X^{(n)} \quad (x_1, \ldots , x_n)
\longmapsto x_1 \ldots x_n$} \quad 
 \end{align*}
The members $x_1 \cdots x_n$s in $X^{(n)}$
are formal products whose order of factor is irrelevant
\footnote{In other words, each member is a multiset of the size $n$.
Cf. (\ref{ordirmult}) in the final subsection}.
On the other hand, the members of $X^{\otimes n}$ are
ordered sequences, hence the map $F$ forgets the order of factor. \\
Note  $\mc{X}^{(n)}$ is automatically a $\sigma$-field over $X^{(n)}$. 
That is, the measurable space $\msn{X}{n}$ is the projection of
the $n$-th direct product $(X^{\otimes n}, \mc{X}^{\otimes n})$
along the map $F: X^{\otimes n} \longrightarrow X^{(n)}$
forgetting the order.

\noindent {\bf Notation:}
An element $x_1 \cdots x_n \in X^{(n)}$ is abbreviated by $\msbf{x}$
when $n$ is clear from the context, while an element 
$(x_1,  \ldots , x_n) \in X^{ \otimes n}$ is abbreviated by 
$\vec{\msbf{x}}$.
}
 \end{defn}

\begin{prop}[$\mc{X}^{(n)}$ as equaliser] \label{eqTK}
In $\opTsKer$, the object $\mc{X}^{(n)}$ becomes the equaliser of the $n!$-symmetries of
$\, \mc{X}^{\otimes n}$. 
The transition kernel $\eq_{\mc{X}} : \mc{X}^{(n)} \longrightarrow \mc{X}^{\otimes n}$
is specified by
\begin{align}  \label{eqTKf1}
\eq_{\mc{X}} (-, (x_1, \ldots, x_n)) = \delta(-, x_1 \cdots x_n)
\end{align}
For any transition kernel $\kappa$ to $\mc{X}^{\otimes n}$
equalising the $n!$-symmetries, its unique factorisation
$\eq_{\mc{X}} \backslash \kappa$ via $\eq_{\mc{X}}$
is given by
\begin{align}  \label{eqTKf2}
(\eq_{\mc{X}} \backslash \kappa) (-, x_1 \cdots x_n) = \kappa(-, (x_1, \ldots , x_n))
\end{align}
$\eq_{\mc{X}} \backslash \kappa$ is well defined as a set map
(independently of the ordering $x_1 \cdots x_n$) 
because $\kappa(-, (x_1, \ldots , x_n)) =
\kappa(-, (x_{\sigma(1)}, \ldots , x_{\sigma(n)}))$
for all $\sigma \in \mathfrak{S}_n$. 
 \end{prop}
\begin{proof}{}
($\eq_{\mc{X}} \backslash \kappa$ is a kernel) 
We show measurability of $\eq_{\mc{X}} \backslash \kappa$
for the second argument as it is obviously  a measure for the first argument. For a fixed measurable $A$,
let $\kappa_A$ and
$(\eq_{\mc{X}} \backslash \kappa)_A$
denote respectively the functions
$\kappa(A, (x_1, \ldots , x_n)):
\mc{X}^{\otimes n} \rightarrow \zeroinf$
and
$(\eq_{\mc{X}} \backslash \kappa) (A, x_1 \cdots x_n):
\mc{X}^{(n)}
\rightarrow \zeroinf$.
As the measurable $\kappa_A$ equalises the $n!$-symmetries,
$\kappa_A^{\mbox{-}1}(V)$
is invariant under the permutations so that  
$\sigma(\kappa_A^{\mbox{-}1}(V))=\kappa_A^{\mbox{-}1}(V) \in \mc{X}^{\otimes n}$
for any $\sigma \in \mathfrak{S}_n$
with any measurable
$V \subset \zeroinf$.
This implies measurability of $(\eq_{\mc{X}} \backslash \kappa)_A$
so that
$(\eq_{\mc{X}} \backslash \kappa)_A^{\mbox{-}1}(V) \in \mc{X}^{(n)}$
as
$( \eq_{\mc{X}} \backslash \kappa)_A^{\mbox{-}1}(V)
=F(\kappa_A^{\mbox{-}1}(V))$
and $\FI(F(\kappa_A^{\mbox{-}1}(V)))
= \bigcup_{\sigma \in \mathfrak{S}_n}
\sigma(\kappa_A^{\mbox{-}1}(V))=
\kappa_A^{\mbox{-}1}(V)$.

\smallskip

\noindent (Uniqueness)
Precompose any transition kernel $\tau$ (of the
codomain $\mc{X}^{(n)}$)
 to $\eq_{\mc{X}}$:\\
$\begin{aligned}
&  \eq_{\mc{X}} \Comp \tau (-, (x_1, \ldots, x_n))
 =\textstyle  \int_{X^{(n)}} \tau (-, \msbf{y})
\, \eq ( d \msbf{y}, (x_1, \ldots, x_n))
= \textstyle \int_{X^{(n)}} \tau (-, \msbf{y})
\, \delta(d \msbf{y}, x_1 \cdots x_n) 
 = \tau (-, x_1 \cdots x_n)
\end{aligned}$ \\
Thus $\eq_{\mc{X}} \backslash \kappa$ gives
the unique factorisation of $\kappa$.
\end{proof}

The equaliser of
Proposition \ref{eqTK} 
acts also on morphisms for 
any $\kappa : \mc{Y} \longrightarrow \mc{X}$:
The transition kernel
$\kappa^{(n)} : \mc{Y}^{(n)} \longrightarrow \mc{X}^{(n)}$
is defined by the unique factorisation via $\mc{X}^{(n)}$ of
$\kappa^{\otimes n } \Comp \eq_{\mc{Y}}$ equalising the $n!$-symmetries
of $\mc{X}^{\otimes n}$: See the diagram;
\begin{align}
\xymatrix{
\mc{X}^{(n)}   \ar@/^18pt/[rrr]^{\eq_{\mc{X}}} &
\ar@{-->}[l]^{\kappa^{(n)}}
\mc{Y}^{(n)}
\ar[r]_{\eq_{\mc{Y}}} & \mc{Y}^{\otimes n} \ar[r]_{\kappa^{\otimes n}} &
\mc{X}^{\otimes n}
} \label{eqC}
\end{align} 
Below in Proposition \ref{actnlim}, the morphism $\kappa^{(n)}$
is described explicitly. 
Let us see a special simple example.
\begin{exam}
{\em 
$\pr_l^{(n)}: (\mc{X}_1 \& \mc{X}_2)^{(n)} \longrightarrow
 (\mc{X}_1)^{(n)}$ is 
 \begin{align}
 \pr_l^{(n)} (-, x_1 \cdots x_n) =
\delta (-, (1, x_1) \cdots  (1, x_n)) \label{pleft}
\end{align}
for the left projection $\pr_l: \mc{X}_1 \& \mc{X}_2 \longrightarrow \mc{X}_1$,
which is defined by $\pr_l (-, x)= \delta(-, (1,x))$.

\smallskip

\noindent (\ref{pleft}) is derived as follows using 
$(\pr_l)^{\otimes n} \Comp \eq_{\mc{X}_1 \& \mc{X}_2}=\eq_{\mc{X}_1} \Comp 
(\pr_l)^{(n)}: (\mc{X}_1 \& \mc{X}_2)^{(n)} \longrightarrow (\mc{X}_1)^{\otimes n}$ by (\ref{eqC}): \\
$\begin{aligned}
& (\pr_l)^{\otimes n} \Comp \eq (-, (x_1, \ldots , x_n))
= \textstyle \int_{ (X_1 \uplus X_2)^{\otimes n}} \eq (-, \vec{\msbf{z}})
\, (\pr_l)^{\otimes n} (d \vec{\msbf{z}}, (x_1, \ldots , x_n))
=  \\
& \textstyle \int_{ (X_1 \uplus X_2)^{\otimes n}} \eq (-,  \vec{\msbf{z}})
\, \delta^{\otimes n} (d \vec{\msbf{z}}, ((1, x_1), \ldots , (1, x_n))) 
  =  
\eq (-, ((1, x_1), \ldots , (1, x_n)))
\stackrel{by (\ref{eqTKf1})}{=}\mbox{LHS of (\ref{pleft})}, 
\end{aligned}$ \\
in which $\eq$ denotes $\eq_{\mc{X}_1 \& \mc{X}_2}$.

\smallskip
\noindent On the other hand, \\
$\begin{aligned}
 \eq_{\mc{X}_1} \Comp 
\pr_l^{(n)} (-, (x_1, \ldots , x_n)) & =
\textstyle \int_{X_1^{(n)}} \pr_l^{(n)} ( - , \msbf{x})
\, \eq_{\mc{X}_1} (d \msbf{x}, (x_1, \ldots , x_n)) \\
& 
\stackrel{by (\ref{eqTKf1})}{=}
\textstyle \int_{X_1^{(n)}} \pr_l^{(n)} ( - , \msbf{x})
\, \delta (d \msbf{x}, x_1 \cdots  x_n) = 
\pr_l^{(n)} (-, x_1 \cdots  x_n) = \mbox{RHS of (\ref{pleft})}
\end{aligned}$

}\end{exam}

\bigskip

Proposition \ref{eqTK} directly makes  
the equaliser $\oble{\mc{X}}{n}$ for E$_\mc{X}$
of Definition \ref{MTTfex}  in $\opTsKer$
definable by
\begin{align*}
\oble{\mc{X}}{n} := \msnw{X}{n}
\end{align*}
Then, we observe that the canonical
$p_{n+1,n}: \msnw{X}{n+1} \longrightarrow
\msnw{X}{n}$ is described in $\opTsKer$ by
\begin{align}
p_{n+1,n} (-, z_1 \cdots z_n ) & 
= \delta(-, z_1 \cdots z_n (2,*)) \label{pn1n}
\end{align}
because LHS of (\ref{pn1n}) becomes  
$\eq \Comp p_{n+1,n} (-, (z_1, \ldots , z_n))$ by (\ref{eqTKf2}), but 
$\eq \Comp p_{n+1,n}  =
( \mc{X}_\bullet^{\otimes n} \otimes \pr_r) \Comp \eq$,
whose RHS is calculated as follows with  
$\mc{X}_\bullet$ abbreviating $\mc{X} \& \mc{I}$:
\begin{align*}
& ( \mc{X}_\bullet^{\otimes n}  \otimes \pr_r) \Comp \eq
\, \, (-, (z_1, \ldots , z_n)) \\
& \textstyle \! \!  = \! \! 
\int_{(X \uplus I)^{\otimes (n+1)}}
\eq(-, (\vec{\msbf{y}}, y_{n+1})) \, \, 
( \mc{X}_\bullet^{\otimes n} \otimes \pr_r) (d (\vec{\msbf{y}}, y_{n+1}) , (z_1,
 \ldots , z_n)) \\
& \tag*{where $(\vec{\msbf{y}}, y_{n+1})=
(y_1, \cdots , y_n, y_{n+1})$}  \\
& \textstyle =
\int_{(X \uplus I)^{\otimes n}}
\int_{X \uplus I}
\eq(-, (\vec{\msbf{y}}, y_{n+1}))
\, \, ( \mc{X}_\bullet^{\otimes n} \otimes \pr_r) ((d \vec{\msbf{y}}, dy_{n+1}), (z_1, \ldots
  , z_n)) \tag*{by Fubini-Tonelli} \\
&  \textstyle  =
\int_{(X \uplus I)^{\otimes n}}
\int_{X \uplus I}
\eq(-, (\vec{\msbf{y}}, y_{n+1}))
 \, \, \mc{X}_\bullet^{\otimes n} (d \vec{\msbf{y}}, (z_1, \ldots
  , z_n)) \, \pr_r (dy_{n+1}, *) \\
& \textstyle =
\int_{(X \uplus I)^{\otimes n}}
\eq(-, (\vec{\msbf{y}}, (2,*)))
 \, \, \mc{X}_\bullet^{\otimes n} (d \vec{\msbf{y}}, (z_1, \ldots , z_n)) \\
& = 
\eq(-, (z_1, \cdots , z_n, (2,*))) = 
\mbox{RHS of (\ref{pn1n})} \tag*{by (\ref{eqTKf1})}
 \end{align*}
Iterating the above (\ref{pn1n}) yields
the following description of 
$p_{m, n} := p_{n+1, n} \Comp \cdots \Comp  p_{m-1, m-2} \Comp p_{m, m-1}
:\msnw{X}{m} \longrightarrow
\msnw{X}{n}$
for natural numbers $m > n$:
\begin{align}
p_{m,n}(-, z_1 \cdots z_n ) 
= \delta(-, z_1 \cdots z_n \overbrace{(2, *) \cdots (2, *)}^{m-n} ) \label{pmn}
\end{align}

\bigskip

Intuitively, the transition kernel $p_{n+1,n}$ may be seen to  
forget the rooted element $(2,*)$.
This leads us to define the limit of the sequential diagram of the (rooted
with $\mc{I}$)
equalisers $\oble{\mc{X}}{n}=(\mc{X} \& \mc{I})^{(n)}$
as the countable infinite products
of the (rootless) equalisers $\mc{X}^{(n)}$.

\begin{defn}[measurable space $! \, \mc{X}$]\label{defexob}{\em 
For any measurable space $\mc{X}$,
the following 
measurable space
of the countable infinite products 
of $\mc{X}^{(k)}$s is denoted by
\begin{align*}
! \, \mc{X} := 
 (\underset{k \in \mathbb{N}}{\mathlarger{\uplus}}
 X^{(k)}, 
\underset{k \in \mathbb{N}} {\mathlarger{\&}} \mc{X}^{(k)}) 
\end{align*}}
 \end{defn}

\bigskip

In order to show this provides the limit,
the following measurable function is prepared.

A function $G_{n,
\infty}: \oble{\mc{A}}{n} \longrightarrow \oble{\mc{A}}{\infty} $ is defined by 
\begin{align}
&  G_{n,\infty}: \oble{A}{n} \longrightarrow  \biguplus_{k \in \mathbb{N}} A^{(k)} 
&  (1,a_1) \cdots (1,a_k) (2, *) \cdots (2,*) \longmapsto
(k, a_1 \cdots a_k) \label{Ginfn} 
\end{align}
Note in (\ref{Ginfn}), each element of the underlying set
of $\msnw{A}{n}$ is written
$(1, a_1) \cdots (1, a_{k}) \underbrace{(2, *) \cdots (2,*)}_{n-k}$ with a
 certain $k\leq n$
such that $a_i \in A$ with $i=1, \ldots, k$. 
On the other hand, 
$(k, a_1 \cdots a_{k})$ designates an element 
from the $k$-th component of 
$\underset{n \in \mathbb{N}}{\mathlarger{\uplus}}
 A^{(n)}$ which is the underlying set of 
$\underset{n \in \mathbb{N}}{\mathlarger{\&}}
 \mc{A}^{(n)}$. 
The function $G_{n,\infty}$ is one to one (but not surjective)
and $(\oble{\mc{A}}{n}, \oble{\mc{A}}{\infty})$-
measurable.  For example, $G_{n,\infty}$ makes
(\ref{pmn}) definable deterministically in terms of the Dirac delta.

\begin{exam} For $n < m$, for any $- \in \oble{\mc{A}}{m}$
and  $\msbf{z} \in (A \uplus I)^{(n)}$,  
$$p_{m,n}(-, \msbf{z}) = \delta (G_{m, \infty}(-), G_{n, \infty}(\msbf{z}))$$
\end{exam}
This example may suggest the following definition
$p_{\infty, n}$, for which $m$ tends to $\infty$, causing
$G_{m, \infty}$ to become the identity intuitively.

\begin{thm}[limit L$_{\mc{X}}$ for $\opTsKer$] \label{limTK}
For any object $\mc{X}$ in $\opTsKer$, 
$! \, \mc{X}$ of Definition \ref{defexob}
with the following kernels $\{ p_{\infty, n} \}_n$,
definable from the measurable functions $G_{n,
\infty}$ and Dirac delta becomes the limit of L$_{\mc{X}}$:
\begin{align}
& \oble{\mc{X}}{\infty} = ! \, \mc{X} \quad \mbox{and} \quad p_{\infty, n} 
\big( -, \msbf{z} \big)
:= \delta \big( -,  G_{n, \infty} (\msbf{z}) \big):
\textstyle \underset{k \in \mathbb{N}}{\bigwith} \mc{X}^{(k)}
 \longrightarrow \oble{\mc{X}}{n}   
 \nonumber \\ 
 & \text{
 To be explicit,} \quad    p_{\infty, n} \big( -, (1, x_1) \cdots (1, x_{k}) (2, *) \cdots (2,*)
 \big)
:=  \delta \big( -,  (k, x_1 \cdots x_{k}) \big) \label{pinf} 
\end{align}
\end{thm}

\begin{proof}{}
It is direct from the definition $p_{\infty, n}$
indexed by $n$ provides a cone. 
Then the proof consists of the following two claims.



\noindent (Claim 1 on factorisation of cone) \\
Any cone $\{ \tau_n \}_n$ factors through the cone $\{ p_{\infty, n} \}_n$ by
the following countable infinite product morphism $\tau$
whose codomain is 
$\underset{n \in \mathbb{N}}{\bigwith} \mc{X}^{(n)}$; 
\begin{align}
& \tau := \underset{k \in \mathbb{N}}{\bigwith} \, \,  \big(
\pr_l^{(k)}
 \Comp \tau_k \big), \, \mbox{where $\pr_l^{(k)} : 
\msnw{X}{k} \longrightarrow  \mc{X}^{(k)}$}, \label{factotau}
\end{align}

See the diagram below for the construction of the mediating $\tau$:
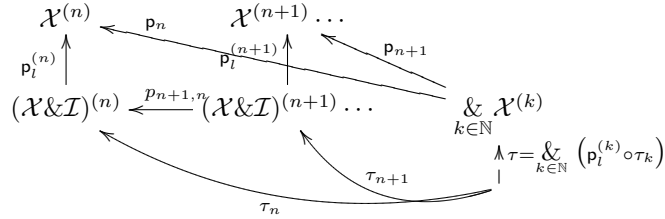
$$
\xymatrix@R=15pt
{   \mc{X}^{(n)} & \mc{X}^{(n+1)} \cdots &    \\
\msnw{X}{n} \ar[u]^{\pr_l^{(n)}}  & \msnw{X}{n+1}  \ar[l]_{p_{n+1,n}}
\ar[u]^(.6){\pr_l^{(n+1)}}  \cdots
&   \underset{k \in \mathbb{N}}{\mathlarger{\&}} \mc{X}^{(k)}
\ar[ull]_(.8){\pr_n}   \ar[ul]_{\pr_{n+1}}
\ar@/^9pt/[ll]^(.7){p_{\infty, n}}   \ar@/^10pt/[l]^{p_{\infty n+1}}  
\\   & & \ar@/^18pt/[ull]^{\tau_n}   \ar@/^18pt/[ul]_{\tau_{n+1}}  
\ar@{-->}[u]_(.4){\tau = \underset{k \in \mathbb{N}}{\mathlarger{\&}} \, \,  
\big(
\pr_l^{(k)}
 \Comp \tau_k \big)}
}
$$
\noindent Proof of Claim 1:
By the definition of $p_{\infty, n}$, it needs to prove for any $n$
and for any $\msbf{z} \in (X \uplus I)^{(n)}$, 
\begin{align*}
\tau (-, G_{n, \infty} (\msbf{z}))  = \tau_n (-,\msbf{z}) 
\end{align*}
Each $\msbf{z} =(1, x_1) \cdots (1, x_{k}) (2, *) \cdots
 (2,*) \in (X \uplus I)^{n}$ for certain $k \leq n$.
Then the following starts with LHS and ends with RHS.
\begin{align*}
\tau(-, (k, x_1 \cdots x_k)) 
& =
\pr_l^{(k)}
 \Comp \tau_k (-, x_1 \cdots x_k ) =  \textstyle 
\int_{(X \uplus I)^{(k)}} 
\tau_k (-, \msbf{y}) \, \pr_l^{(k)} (d \msbf{y}, x_1 \cdots x_k) \\
& =
\textstyle 
\int_{(X \uplus I)^{(k)}} 
\tau_k (-, \msbf{y}) \, \delta (d \msbf{y}, (1, x_1) \cdots (1, x_k)) 
\tag*{by (\ref{pleft})} \\
& = 
\tau_k (-, (1, x_1) \cdots (1, x_k)) \\ & 
 = \textstyle 
\int_{(X \uplus I)^{(n)}}  
\tau_n (-, \msbf{y}) 
\, p_{n, k} (d \msbf{y}, 
(1, x_1) \cdots (1, x_k)) \tag*{
as $\tau_k = p_{n, k} \Comp \tau_n$} \\
& =
\textstyle 
\int_{(X \uplus I)^{(n)}}  
\tau_n (-, \msbf{y}) 
\, \delta (d \msbf{y}, 
(1, x_1) \cdots (1, x_k) (2, *) \cdots (2, *)) \tag*{by (\ref{pmn})} \\
& = \tau_n (-, (1, x_1) \cdots (1, x_k) (2, *) \cdots (2, *)) 
\end{align*}

\smallskip

\noindent (Claim 2 on uniqueness of factorisation)
For any $n$, 
\begin{align}
  \pr_l^{(n)} \Comp p_{\infty, n} = \pr_n, \label{lik}
\end{align}
where $\pr_n$ is the $n$-th projection of the product
$\underset{n \in \mathbb{N}}{\bigwith} \mc{X}^{(n)}$.
See again the above diagram.

\noindent Proof of Claim 2:
By the following starting from 
$LHS(-, x_1 \cdots x_n)$ and ending with
$RHS(-, x_1 \cdots x_n )$ 
\begin{align*}
 \textstyle
\int_{{(X \uplus I)^{(n)}}} 
p_{\infty, n} (-, \msbf{y})
\, \pr_l^{(n)} (d \msbf{y}, x_1 \cdots x_n)   
& = \textstyle 
\int_{{(X \uplus I)^{(n)}}} 
p_{\infty, n} (-,  \msbf{y})
\, \delta (d \msbf{y}, (1, x_1) \cdots (1, x_n)) \tag*{by (\ref{pleft})} \\
& = 
p_{\infty, n} (-,  (1, x_1) \cdots (1, x_n))
 =  \delta(-, (n, x_1 \cdots x_n)) \tag*{by (\ref{pinf})} 
\end{align*}

\smallskip

\noindent Claim 2 guarantees uniqueness of the factorisation:
Given any factorisation $\tau'$ such that $p_{\infty, n} \Comp \tau'= \tau_n$
of Clam 1,
composing $\pr_l^{(n)}$ to the equality yields 
$\pr_n \Comp \tau'= \pr_l^{(n)} \Comp  \tau_n$ by Clam 2.
Hence by the universality of the product,
$\tau' =  \underset{n \in \mathbb{N}}{\mathlarger{\&}}
(\pr_n \Comp \tau') = 
\underset{n \in \mathbb{N}}{\mathlarger{\&}}
 (\pr_l^{(n)} \Comp  \tau_n )= 
\tau$.
\end{proof}

The action of the limit of Definition \ref{actmor} 
is described concretely in $\opTsKer$:
\begin{prop}[morphisms $\kappa^{(n)}$ and $\oble{\kappa}{\infty}$ ] \label{actnlim}
Let $\kappa: \mc{Y} \longrightarrow \mc{X}$ in $\opTsKer$. 
\begin{itemize}
 \item[(i)] The morphism $\kappa^{(n)}: \mc{Y}^{(n)} \longrightarrow 
\mc{X}^{(n)}$ of (\ref{eqC}) is described explicitly as follows
for $x_1 \cdots x_n \in X^{(n)}$:
\begin{align*}
\kappa^{(n)} (-, x_1 \cdots x_n):=
    \kappa^{\otimes n} (\FI(-), (x_1, \ldots , x_n))
\end{align*} 
This is well defined independently of any enumeration of
the unordered product.
This in particular by Definition \ref{actmor} stipulates $(\kappa \&
	    \mc{I})^{(n)}=\oble{\kappa}{n}$.

\item[(ii)] 
The unique morphisms $\oble{\kappa}{\infty}: \,  \oble{\mc{Y}}{\infty} 
\longrightarrow \oble{\mc{X}}{\infty}$
is explicitly described as follows:
\begin{align*}
 & \oble{\kappa}{\infty} =   \underset{n \in \mathbb{N}}{\mathlarger{\&}}
\, (\kappa^{(n)} \Comp \pr_n) :
\quad  \oble{\mc{Y}}{\infty}  \longrightarrow 
\underset{n \in \mathbb{N}}{\mathlarger{\&}} \mc{X}^{(n)}  = \oble{\mc{X}}{\infty},
\end{align*}
\noindent in which 
$\xymatrix{\kappa^{(n)} \Comp
 \pr_n:
\oble{\mc{Y}}{\infty} = \underset{k \in \mathbb{N}}{\mathlarger{\&}} \mc{Y}^{(k)}
 \ar[r]^(.7){\pr_n} & \mc{Y}^{(n)} \ar[r]^{\kappa^{(n)}}
& \mc{X}^{(n)}}$.
\end{itemize}
\end{prop}
\begin{proof}
(i)
\begin{align*}
  \kappa^{\otimes n} \Comp \eq_{\mc{Y}} (-, (x_1, \ldots , x_n)) 
& = \textstyle  \int_{Y^{\otimes n}} \eq_{\mc{Y}} (-, (y_1, \ldots , y_n))
\, \kappa^{\otimes n} (d (y_1, \ldots , y_n), (x_1, \ldots , x_n)) 
\tag*{by (\ref{eqTKf1})} \\
& = \textstyle  \int_{Y^{\otimes n}} \delta  (-, y_1 \cdots y_n)
\, \kappa^{\otimes n} (d (y_1, \ldots , y_n), (x_1, \ldots , x_n)) \\
& = \textstyle  \int_{Y^{\otimes n}} \delta  ( \FI(-), (y_1, \ldots , y_n))
\, \kappa^{\otimes n} (d (y_1, \ldots , y_n), (x_1, \ldots , x_n)) \\
& = \kappa^{\otimes n} (\FI (-), (x_1, \ldots , x_n))
\end{align*}
In the 3rd line, 
$\delta  ( \FI(-), (y_1, \ldots , y_n))$
replaces $\delta  (-, y_1 \cdots y_n)$ equivalently.

\bigskip

\noindent (ii)
$\oble{\kappa}{\infty}$ by Definition \ref{actmor} is 
the unique factorisation of the cone $\{
 (\kappa \cdot \pr_l \,  \&  \, \mc{I} \cdot \pr_r)^{(n)} \Comp  p_{\infty, n} \}_n$.
Hence by (\ref{factotau}), $\oble{\kappa}{\infty}$ is
\begin{align*}
& \underset{n \in \mathbb{N}}{\mathlarger{\&}} \, 
\big( \pr_l^{(n)} \Comp
( \kappa \cdot \pr_l \,  \&  \, \mc{I} \cdot \pr_r)^{(n)} \Comp
p_{\infty, n} 
 \big) 
 = \underset{n \in \mathbb{N}}{\mathlarger{\&}} \, 
  (\kappa^{(n)} \Comp \pr_n)
\end{align*}
The equality is by the commutativity of the following diagram,
in which the right square commutes by the functoriality of
$(-)^{(n)}$ and by $\pr_l \Comp (\kappa \cdot \pr_l \,  \&  \, \mc{I} \cdot
 \pr_r) = \kappa \Comp \pr_l$, and the left triangle does
by (\ref{lik} ).

$\xymatrix@R=15pt
{& \mc{Y}^{(n)} \ar[rr]^{\kappa^{(n)}} & &  \mc{X}^{(n)}
\\ ! \mc{Y} \ar[r]^{p_{\infty, n}} \ar[ur]^{\pr_n} &
\msnw{Y}{n} \ar[u]^{\pr_l^{(n)}}
\ar[rr]^{(\kappa \cdot \pr_l \& \mc{I} \cdot \pr_r)^{(n)}}
& & \msnw{X}{n} \ar[u]_{\pr_l^{(n)}}
}$

\end{proof}

\bigskip

In $\opTsKer$, Fubini-Tonelli Theorem guarantees the distribution
properties for the free exponential.

\begin{prop}[distribution of $\otimes$ 
over the equalisers and the limits in $\opTsKer$] \label{tencMS}
For any measurable space $\mc{Z}$, the following holds in
$\opTsKer$.
\begin{enumerate}
 \item[(i)]
The monoidal product distributes over the equaliser E$_{\mc{X}}$ so that 
$ ( \eq_{\mc{X}} \otimes \mc{Z}, \mc{X}^{(n)} \otimes \mc{Z})$
becomes the equaliser of ($n!$-symmetries) $\otimes \mc{Z}$ of
$\mc{X}^{\otimes n} \otimes \mc{Z}$. Thus
for any transition kernel $\kappa$ to $\mc{X}^{\otimes n} \otimes \mc{Z}$
equalising the ($n!$-symmetries) $\otimes \mc{Z}$, its unique factorisation
$(\eq_{\mc{X}} \otimes \mc{Z}) \backslash \kappa$ via $\eq_{\mc{X}} \otimes \mc{Z}$
is given by
\begin{align}
(\eq_{\mc{X}} \otimes \mc{Z}) \backslash \kappa (-, (x_1 \cdots x_n, z)) = \kappa(-, (x_1, \ldots , x_n, z))
\end{align}

\item[(ii)]
The monoidal product distributes over the limit 
\mbox{L}$_{\mc{X}}$.
\end{enumerate}
\end{prop}
\begin{proof}
A direct generalisation of the respective proofs
of Proposition \ref{eqTK} and 
of Theorem \ref{limTK} by 
consistently
replacing $\int_{(-)}$ with double integration  $\int_{(-) \times Z}
=\int_{(-)} \int_{Z}$ by Fubini-Tonelli of Proposition \ref{FTsfin}.

\bigskip

\noindent (i)
See Appendix \ref{proofdist} for the proof.

\smallskip

\noindent (ii) 
The two claims in the proof of Theorem \ref{limTK}
are directly generalised as follows:



\noindent (Claim 1) 
Any cone $\{ \tau_n \}_n$ to the sequential diagram 
$\{ p_{n,n+1} \otimes \mc{Z} \}_n$ 
factors by
the following countable infinite product morphism $\tau$:
\begin{align}
\tau :=\underset{k \in \mathbb{N}}{\mathlarger{\&}} 
(\pr_l^{(k)}  \Comp \tau_k ) : \, \,  \oble{\mc{Y}}{\infty} \longrightarrow
\underset{k \in \mathbb{N}}{\mathlarger{\&}} (\mc{X}^{(k)} \otimes \mc{Z})
\label{genfactotau} 
\end{align} 
See Appendix \ref{proofdist} for the proof of Claim 1.

\bigskip

\noindent (Claim 2)
The claim $(\pr_l^{(n)} \otimes \mc{Z}) \Comp (p_{\infty, n} \otimes \mc{Z}) = \pr_n \otimes \mc{Z}$ for any $n$ is direct from the original one.

\smallskip
Hence, (\ref{genfactotau}) gives the mediating morphism of the cone
via the following isomorphism of its codomain:
\begin{align*}
\underset{k \in \mathbb{N}}{\mathlarger{\&}} (\mc{X}^{(k)} \otimes
 \mc{Z}) \cong 
(\underset{k \in \mathbb{N}}{\mathlarger{\&}} \mc{X}^{(k)}) \otimes
 \mc{Z} & \quad \quad 
(k, (x_1 \cdots x_k, z)) \longmapsto ((k,(x_1 \cdots x_k)), z) 
\end{align*}
\end{proof}

\bigskip

By Propositions \ref{eqTK} and \ref{tencMS} and  
Theorem \ref{limTK};
\begin{thm} \label{ftsect3} 
The monoidal category $\opTsKer$ with biproducts
has the free exponential.
\end{thm}

\section{\large The Orthogonality Category $\T{\mc{I}}{\opTsKer}$} \label{sect4}
This section starts with presenting
a focused orthogonality in $\opTsKer$
in terms of Lebesgue integral.
This in particular allows a categorical reformulation of
the variable change of integral along push forward.
Accordingly, a general categorical construction of
Section \ref{sect1} is shown to apply to the free exponential of
Section \ref{sect3}
in order to obtain the free exponential in $\T{\mc{I}}{\opTsKer}$,
which is a main goal Corollary \ref{corconhold}
in Section \ref{sect4.2}. 
The orthogonality is shown in Sections
\ref{sect4diseq} and \ref{sect4.2} to satisfy
the conditions formulated in
Section \ref{sect1.3} on distribution of monoidal product
respectively over
the equalisers and the limit.
Section \ref{sect4.3} characterises concretely the equalisers
within $\T{\mc{I}}{\opTsKer}$ and 
Section \ref{sect4.5} follows to characterise their limit.
Section \ref{sect4.4} shows that 
our free exponential 
can be considered as a continuous extension
of the exponential of $\Pcoh$. 

In the sequel, the object $J$ for the orthogonality (Definition \ref{defort}) is the monoidal unit $\mc{I}$.

\subsection{\normalsize Focused Orthogonality between Measures and
Measurable Functions in $\opTsKer$}
\begin{defn}[inner product]\label{inprod}{\em
For a measure $\mu \in \opTKer(\mc{X}, \mc{I})$ and a measurable function $f \in
\opTKer(\mc{I}, \mc{X})$, we define
$$\begin{aligned}
\inpro{f}{\mc{X}}{\mu} := \int_{X} f d\mu 
\end{aligned}$$
}
\end{defn}

\noindent Then the two operators in Definition 
\ref{opekernel} become characterised as follows:
 \begin{prop}[reciprocity between $\kappa^*$ and $\kappa_*$] \label{adjn}
In $\opTKersfin$, for any 
measure $\mu: \mc{X} \longrightarrow \mathcal{I}$,
any measurable function
$f: \mc{I} \longrightarrow \mathcal{Y}$ and 
any transition kernel $\kappa: \mc{Y} 
\longrightarrow \mc{X}$,  
\begin{align*}
& \inpro{f}{\mc{Y}}{\kappa_* \mu }
   = \inpro{\kappa^*  f}{\mc{X}}{\, \mu}
& \mbox{\em Equivalently,} \quad  \inpro{f}{\mc{Y}}{\mu \Comp \kappa}
    = \inpro{\kappa \Comp f}{\mc{X}}{\, \mu}
\end{align*}
\end{prop}
\begin{proof}
The following starts from LHS and ends with RHS of the assertion,
using Fubini-Tonelli (Proposition \ref{FTsfin}): 
$\begin{aligned}
& \textstyle \int_{Y} f (y) (\kappa_* \mu) (dy) 
=
\int_Y f(y) \int_X \kappa(dy, x) \mu(dx)
 \textstyle = 
\int_X \mu(dx)  \int_Y f(y) \kappa(dy,x) 
=  \int_{X} (\kappa^* f)(x) \mu(dx)
\end{aligned}$
\end{proof}

\begin{defn}[orthogonality in terms of integral] 
For a measurable function $f \in
\opTKersfin(\mc{I}, \mc{X})$ and 
a measure $\mu \in \opTKersfin(\mc{X}, \mc{I})$, the relation $\bot_{\mc{X}}
\subset \opTKersfin(\mc{I},\mc{X}) \times \opTKersfin(\mc{X},\mc{I})$ is defined
\begin{align}
\ort{f}{\mathcal{X}}{\mu} \quad \text{if and only if} \quad 
\inpro{f}{\mc{X}}{\mu} \, \, \leq 1  \label{ortint}
 \end{align}
\end{defn}
\begin{lem}
The relation (\ref{ortint}) gives a focused orthogonality
in $\opTKersfin$. 
\end{lem}
\begin{proof}
By Proposition \ref{adjn}.
\end{proof}

\subsection{\normalsize Push Forward Integral as Reciprocity} \label{sect.4.1.1}
The measure-theoretic  variable variable change along push forward
in Definition \ref{pfm}
is reformulated category theoretically in $\opTsKer$
in terms of the reciprocity of orthogonality.
The category theory first avoids the measure-theoretic
abuse of the notation $\mu \Comp F^{-1}$
for the push forward measure, which measure is obtained in 
$\opTsKer$ by precomposing a certain morphism
to a measure $\mu: \mc{Y} \longrightarrow \mc{I}$.
In this subsection two kinds of variable changes of integrals
over the push forward measures are shown to be 
characterised  in terms of the reciprocity of the categorical morphims
$\eq$ and $p_{\infty, n}$, respectively.
Recall that the two were the main ingredients in previous Section \ref{sect3}
to construct the equaliser and the limit in $\opTsKer$.

\smallskip

First it is direct to describe the push forward along
the measurable order forgetting map $F$ of Definition \ref{ms(n)}
using the categorical morphism $\eq$.
\begin{prop}[variable change along the push forward
$F :  \mc{X}^{\otimes n} \longrightarrow \mc{X}^{(n)}$
as reciprocity of $\eq$]
{\em The variable change property along the push forward $F : \mc{X}^{\otimes n}
\longrightarrow \mc{X}^{(n)}$ for Definition \ref{ms(n)} 
is reformulated in terms of 
the categorical composition and precomposition of
$\eq: \mc{X}^{(n)} \longrightarrow  \mc{X}^{\otimes n}$ in $\opTsKer$ as
follows:
\begin{align}
\int_{X^{(n)}} f \, \, d (\mu \Comp \eq)
= 
\int_{X^{\otimes n}} (\eq \Comp f) \, \, d \mu \label{pfmpcat}
\end{align}
The reformulation is obtained because the push forward measure $\mu \Comp F^{-1}$ (resp. the measurable function $f \Comp
 F$) in (\ref{PFM}) in Definition \ref{pfm} is
$\mu \Comp \eq$ (resp. $\eq \Comp f$) when
any measure $\mu$ (resp. measurable function $f$) is seen as 
a morphism $\mc{X}^{\otimes n} \longrightarrow \mc{I}$
(resp. $ \mc{I} \longrightarrow \mc{X}^{\otimes n}$) in $\opTsKer$. \\
(\ref{pfmpcat}) is obviously a reciprocity of the orthogonality for $\opTsKer$
as the equality is 
$$\inpro{f}{X^{(n)}}{\mu \Comp \eq} = 
\inpro{ \eq \Comp f}{X^{\otimes n}}{\mu}$$

Alternatively putting $f=\eq \backslash g$ with a measurable
 function $g$ on $\mc{X}^{\otimes n}$ equalising the $n!$ symmetries,
\begin{align*}
\int_{X^{(n)}} (\eq \backslash g) \, \, d (\mu \Comp \eq)
= 
\int_{X^{\otimes n}} g \, \, d \mu 
\end{align*}
}\end{prop}

\bigskip

Second, 
the push forward along the measurable function
$G_{n, \infty}$ defined by (\ref{Ginfn} )
is described by a
reciprocity of the categorical morphism  
$p_{\infty, n}$ defined in Theorem \ref{limTK}.

\begin{prop}[variable change along
$G_{n, \infty} \times \mc{B}: \oble{\mc{A}}{n} \otimes \mc{B} \longrightarrow 
\oble{\mc{A}}{\infty} \otimes \mc{B}$
as reciprocity for $p_{\infty, n} \otimes
\mc{B}$] \label{recippB}
For any measurable function
$f : \mc{I} \longrightarrow \oble{\mc{A}}{\infty} \otimes \mc{B}$ 
and any measure $\mu : \oble{\mc{A}}{n} \otimes \mc{B}  \longrightarrow \mc{I}$ in $\opTsKer$, the variable change property along the 
push forward $G_{n,\infty} \times B$ 
\begin{align*}
\textstyle 
\int_{\oble{A}{\infty} \times B} f ( \msbf{y})\, \, (\mu \Comp (G_{n,\infty}
 \times B)^{-1}) (d \msbf{y}) =
\textstyle
\int_{\oble{A}{n} \times B}  f ( (G_{n,\infty} \times B) (\msbf{y}')) \, \,
\mu (d \msbf{y}')
\end{align*}
is the reciprocity in $\opTsKer$  
\begin{align*}
& \int_{\oble{A}{\infty} \times B} f \, d (\mu \Comp (p_{\infty, n}
 \otimes \mc{B})) 
= \int_{\oble{A}{n} \times B}  ( (p_{\infty, n} \otimes \mc{B} ) \Comp f) \, d \mu \end{align*}
\end{prop}
\begin{proof}
The assertion is direct by the following (i) and (ii) respectively 
on composing and on precomposing with $p_{\infty, n} \otimes \mc{B}$: 

\noindent (i)
For any measurable function 
$f : \mc{I} \longrightarrow \oble{\mc{A}}{\infty} \otimes \mc{B}$ in $\opTsKer$,
\begin{align*}
((p_{\infty, n} \otimes \mc{B}) \Comp f) (I, \msbf{y})
= f(I,  (G_{n, \infty} \times B) (\msbf{y}))
\quad \mbox{for $\msbf{y} \in \oble{A}{n} \times B$}
\end{align*}

\noindent (ii) 
For any measure
$\mu : \oble{\mc{A}}{n} \otimes \mc{B}  \longrightarrow \mc{I}$ in $\opTsKer$,
\begin{align*}
\mu \Comp (p_{\infty, n} \otimes \mc{B}) = \mu \Comp (G_{n, \infty}
 \times B)^{-1} 
\end{align*}
That is,  the composition $\mu \Comp (p_{\infty,n} \otimes \mc{B})$ is the push forward
	    measure
of $\mu$ along 
	    $G_{n,\infty} \times B$.

 \noindent (i) is direct by the definition of
$ p_{\infty, n}$ in terms of $G_{n,\infty}$ in Theorem \ref{defexob}.

\noindent (ii) holds by the following whose
third equality is
by variable change along $G_{n, \infty} \times \mc{B}$.
\begin{align*}
\textstyle (\mu \Comp (p_{\infty, n} \otimes \mc{B})
) (-, *) 
& \textstyle  =
\int_{\oble{A}{n} \times B} 
(p_{\infty, n} \otimes \mc{B}) (-, (\msbf{x},y)) \, 
\mu(d (\msbf{x}, y), *)  \\
& 
\textstyle = 
\int_{\oble{A}{n} \times B} \delta (-, (G_{n,  \infty} \times B) (\msbf{x},y)) \, 
\mu(d (\msbf{x},y), *) 
\\ & 
\textstyle = \int_{\oble{A}{\infty} \times B} \delta (-, (\msbf{x}',y)) \, 
\mu( (G_{n,  \infty} \times B)^{-1} (d (\msbf{x}',y)), *) =  
\mu ( (G_{n,  \infty} \times B)^{-1}(-), *)
\end{align*}
\end{proof}

\subsection{\normalsize
Distribution of Monoidal Product over Equalisers in
$\T{\mc{I}}{\opTsKer}$} \label{sect4diseq}
The goal of this subsection is to show the distributivity of
the tensor product over the equalisers E$_{\msbf{A}}$ in
$\T{I}{\cC}$ when $\cC=\opTsKer$. The main technical ingredient 
is Proposition \ref{invexis}. This is preceded by introducing barycenters 
in Definition \ref{defbary}, which will also become crucial in the next
subsection \ref{sect4.3} to characterise the equalisers in the orthogonality
category over our measure theoretic framework.
 
We start with defining that the equaliser 
$A^{(n)}$ with $\cC=\opTsKer$ of Proposition \ref{eqTK}
lifts to that $\msbf{A}^{(n)}$ in $\T{\mc{I}}{\opTsKer}$,
the same as Proposition \ref{PropequalnT}.
\begin{defn}[equaliser $\msbf{A}^{(n)}$ in $\T{\mc{I}}{\opTsKer}$] 
\label{genequalnT}
{\em In $\T{\mc{I}}{\opTsKer}$ for every object $\msbf{A}$, 
the following object $\msbf{A}^{(n)}$ with $\eq$ of  Proposition \ref{eqTK}
becomes the equaliser of
the $n!$-symmetries of $\msbf{A}^{\otimes n}$:
\begin{align}
& \msbf{A}^{(n)} = (A^{(n)}, (\msbf{A}^{(n)})_p), \quad
 \mbox{where} \nonumber  \\
(\msbf{A}^{(n)})_p
& :=\ccrc{ 
\left\{
  \eq \backslash h \mid h \in 
(\msbf{A}^{\otimes n})_p
\, \,  \mbox{equalises the $n!$-} 
\mbox{symmetries  of $\msbf{A}^{\otimes n}$} 
\right\}
} \label{rtleequalnT} 
\end{align}
Obviously the definition is general enough to subsume
$(\oble{\msbf{A}}{n})_p$ of (\ref{equalnT})
in  Proposition \ref{PropequalnT}
when $\msbf{A}$ is instantiated in particular with $\msbf{A} \& \msbf{I}$.
}\end{defn}

\smallskip

\begin{defn}[barycenter $s_n(g)$ as composing an endomorphism $s_n$
on $\mc{A}^{\otimes n}$ ] \label{defbary}
{\em
In the category $\opTsKer$,
the $n$-th barycenter $s_n$ is defined as the 
following endomorphism on $\mc{A}^{\otimes n}$:
\begin{align}
 s_n(-, (a_1, \ldots , a_n))=
\frac{1}{n!} \sum_{\sigma \in \mathfrak{S}_n}
\delta (-, (a_{\sigma(1)}, \ldots , a_{\sigma(n)}))
 \label{formdefbary}
\end{align}
For a measurable function $g: \mc{I} \longrightarrow  \mc{A}^{\otimes
 n}$ in $\opTsKer$, its barycenter
$s_n(g): \mc{I}  \longrightarrow  \mc{A}^{\otimes n}$
is defined to be the categorical composition $s_n \Comp g$ in $\opTsKer$:
 \begin{align*}
s_n (g) &    := s_n \Comp g = 
\frac{1}{n!} \sum_{\sigma \in \mathfrak{S}_n} \sigma (
 g) \tag*{where $\sigma (g) (a_1, \ldots , a_n)
= g (a_{\sigma(1)}, 
\ldots , a_{\sigma(n)})$}  
\end{align*}
In particular when putting $g=f_1 \otimes \cdots \otimes f_n$
of the domain $\mc{I}^{\otimes n} \cong \mc{I}$
with $f_i : \mc{I} \longrightarrow \mc{A}$ for $i=1, \ldots , n$;
\begin{align*}
s_n (f_1 \otimes \cdots \otimes f_n)
=
\frac{1}{n!} \sum_{\sigma \in \mathfrak{S}_n} f_{\sigma(1)} \otimes
 \cdots \otimes f_{\sigma(n)}
\end{align*}
It is direct that $(s_n (g)) (a_1, \ldots , a_n)
=s_n \Comp g \,  (I, (a_1, \ldots , a_n)) =
\frac{1}{n!} \sum_{\sigma \in \mathfrak{S}_n} g(I, (a_{\sigma(1)}, 
\ldots , a_{\sigma(n)}))$. 
}\end{defn}

\begin{rem}{\em 
The barycenters characterise the invariant morphisms: 
$s_n(g)=g$ if and only if $g$ equalises 
the $n!$-symmetries of $\mc{A}^{\otimes n}$.
Moreover, in the orthogonality category, as $s_n$ is an endomorphism on 
$\msbf{A}^{\otimes n}$, 
$s_n(g)$ equalises $n!$-symmetries of $\msbf{A}^{\otimes n}$
for any $g \in (\msbf{A}^{\otimes n})_p$, hence 
$\eq \backslash s_n(g) \in \msbf{A}^{(n)}$.
In the next subsection it is shown in Proposition \ref{corgene}
that the barycenters generate $(\msbf{A}^{(n)})_p$
in order to characterise $\msbf{A}^{(n)}$.  
}\end{rem}

\smallskip

\begin{prop}[Each equaliser in $\opTsKer$ has a left inverse] \label{invexis}
In $\opTsKer$ for any object $\mc{X}$,
the equaliser $\eq_{\mc{X}}$ has a left inverse 
$\eq^{\flat}_{\mc{X}}$
defined to be the unique factorisation of the $n$-th barycenter $s_n$
of Definition \ref{defbary}.
\begin{align}
\eq^{\flat}_{\mc{X}} :=
\eq_{\mc{X}} \backslash s_n : \mc{X}^{\otimes n} \longrightarrow  \mc{X}^{(n)} 
\label{definvexis}
\end{align}
\end{prop}
\begin{proof}{}
\begin{align*}
\eq^{\flat}_{\mc{X}} \Comp \eq_{\mc{X}} (-, x_1 \cdots x_n) & = \textstyle  \int_{X^{\otimes n}} \eq_{\mc{X}}(-, 
(y_1, \ldots , y_n)) \,  \eq^{\flat}_{\mc{X}}(d (y_1, \ldots, y_n), x_1 \cdots x_n)  \\
& =
\textstyle \int_{X^{\otimes n}} 
\,  
\eq_{\mc{X}}(-, 
(y_1, \ldots , y_n)) \,
\frac{1}{n!} 
\sum_{\sigma \in \mathfrak{S}_n}
\delta( d (y_1, \ldots, y_n), (x_{\sigma(1)}, \ldots, x_{\sigma(n)})) 
\tag*{by (\ref{formdefbary}) 
and (\ref{eqTKf2})} \\
& = 
\textstyle \frac{1}{n!} 
\sum_{\sigma \in \mathfrak{S}_n}
 \int_{X^{\otimes n}} 
\,  
\eq_{\mc{X}}(-, 
(y_1, \ldots , y_n)) \,
\delta( d (y_1, \ldots, y_n), (x_{\sigma(1)}, \ldots, x_{\sigma(n)})) \\
& =
\textstyle
\frac{1}{n!} 
\sum_{\sigma \in \mathfrak{S}_n}
\eq_{\mc{X}}(-,(x_{\sigma(1)}, \ldots, x_{\sigma(n)})) \\
&  
= \operatorname{Id}_{\mc{X}^{(n)}} \tag*{
by (\ref{eqTKf1}) as each summand is equal by the equaliser $\eq_{\mc{X}}$
} 
 \end{align*} 
\end{proof}

\begin{rem}{\em 
To be explicit, 
$\eq^{\flat}_{\mc{X}}$ of
(\ref{definvexis}) is descried
for each $x_1 \cdots x_n  \in X^{(n)}$ with any $n$ by  
\begin{align*}
\eq^{\flat}_{\mc{X}} (-, x_1 \cdots x_n) =
\frac{|S_{\vec{\msbf{x}}}|}{n!}
\sum_{\sigma \in \mathfrak{S}_n / S_{\vec{\msbf{x}}} }
\delta(-, (x_{\sigma(1)}, \ldots,  x_{\sigma(n)})),  
\end{align*}
where 
$S_{\vec{\msbf{x}}}:=\{ \sigma \in \mathfrak{S}_n \mid 
(x_1, \ldots , x_n) = (x_{\sigma(1)}, \ldots , x_{\sigma(n)})
\}$ is the stabiliser subgroup fixing
 $\vec{\msbf{x}}=(x_1, \ldots ,
x_n)$. \\
In particular, the explicit description yields the measurability of
the function
mapping every finite multiset $\msbf{x} \in \mc{X}_e$ to 
the cardinality $n! / |S_{\vec{\msbf{x}}}|$ 
of all the enumeration of
$\msbf{x}$.
}\end{rem}

\smallskip

\begin{prop} \label{PROPconditiontnteneq}
The condition (\ref{conditiontnteneq}) holds in $\T{J}{\opTsKer}$
so that $\otimes$ distributes over  E$_{\msbf{A}}$.
\end{prop}
\begin{proof}
By Proposition \ref{invexis} and Example \ref{excondiT}.
\end{proof}

\smallskip

\subsection{\normalsize Characterising Equalisers
$\oble{\msbf{A}}{n}$ in
$\T{\mc{I}}{\opTsKer}$}  \label{sect4.3}
This subsection is concerned with characterising the equaliser
$\oble{\msbf{A}}{n}$ abstractly defined in Proposition \ref{PropequalnT}
in terms of generators
for double orthogonal in $\T{\mc{I}}{\opTsKer}$
(Proposition \ref{corgene}). 
The barycenter construction by
Crubill\'{e} et al \cite{Crubille} 
is directly applied to our categorical framework $\opTsKer$.
The characterisation will be also used in Section
\ref{sect4.2} to show the distributivity of
the tensor product over the limit in $\T{I}{\cC}$.

\smallskip

The following lemma is prepared for
Proposition \ref{corgene}.

\begin{lem} \label{gencomcc}
\begin{itemize}
 \item[(i)]
For any homset $V \subseteq
\opTsKer(\mc{A}^{\otimes n}, \mc{I)}$,
\begin{align*}
 \crc{\{ \eq_{\mc{A}} \backslash s_n (g) \mid g \in \ccrc{V} 
   \}}
=
 \crc{\{ \eq_{\mc{A}g} \backslash s_n (g) \mid g \in V 
   \}}
\end{align*}
\item[(ii)]
For any homset $U \subset \opTsKer(I, \mathcal{X})$ and
any $g \in \cC(I, \mathcal{X})$ (i.e., 
a measurable
function $g$ on $\mathcal{X}$), 
\begin{align*}
g \leq \exists \,   g' \in \ccrc{U}
\Longrightarrow g \in \ccrc{U},
\end{align*}
in which the order $\leq$ is pointwise order between measurable
 functions.
\end{itemize}
\end{lem}

\begin{proof}
(i) We prove ($\supset$) as the converse is tautological.
Take $y$ from RHS, which means
$\ort{y}{\mc{A}^{(n)}}{\eq \backslash (s_n \Comp V)}$,
but $y= y \Comp \eq^{\flat}_{\mc{A}} \Comp
\eq_{\mc{A}}$ by Proposition \ref{invexis}.
Then by reciprocity, 
$\ort{y \Comp \eq^{\flat}_{\mc{A}}}{\mc{A}^{\otimes n}}{
\eq_{\mc{A}} \Comp \eq_{\mc{A}} \backslash
(s_n \Comp V) =s_n \Comp V}$
iff by reciprocity
$\ort{y \Comp \eq^{\flat}_{\mc{A}} \Comp s_n 
}{\mc{A}^{\otimes n}}{V}$
iff 
$\ort{y \Comp \eq^{\flat}_{\mc{A}} \Comp s_n 
}{\mc{A}^{\otimes n}}{\ccrc{V}}$
by $\crc{V}=V^{\circ \circ \circ }$
(i.e., $\ort{r}{}{V}$ iff $\ort{r}{}{\ccrc{V}}$).
Now the same reciprocities applied back,
$\ort{y \Comp \eq^{\flat}_{\mc{A}}}{\mc{A}^{\otimes n}}{s_n \Comp  \ccrc{V}}
= \eq_{\mc{A}} \Comp (\eq_{\mc{A}} \backslash (s_n \Comp  \ccrc{V}))$
as the barycenters are invariant under the symmetries,
iff
$\ort{y= y \Comp \eq^{\flat}_{\mc{A}} \Comp
\eq_{\mc{A}}}{\mc{A}^{(n)}}{\eq_{\mc{A}} \backslash (s_n \Comp  \ccrc{V})}$.

\noindent (ii)
Obvious:
$\forall \nu \in \crc{U}$ $\int_X g \, d \nu \leq \int_X g' \, d \nu \leq 1$.
Thus $g \in \ccrc{U}$.
\end{proof}

\smallskip

We are ready to characterise the equalisers in 
$\T{I}{\opTsKer}$. In what follows, when $h$ has a domain $I^{\otimes n}$,
the domain $I^{(k)}$ of
$\eq \backslash h$ is identified with $I \cong I^{(k)}$.

\begin{prop}[barycentric generators for $(\msbf{A}^{(n)})_p$
and for $(\oble{\msbf{A}}{n})_p$] \label{corgene}
In $\T{I}{\opTsKer}$, 
\begin{itemize}
 \item[(i)]
(generators for the equaliser $\msbf{A}^{(n)}$)
\begin{align*}
(\msbf{A}^{(n)})_p
& =
\ccrc{\{ \eq \backslash s_n(g) \mid g \in 
\ccrc{(\msbf{A}_p^{\otimes n})}  \}}  =
\ccrc{\{ \eq \backslash s_n(g) \mid g \in 
\msbf{A}_p^{\otimes n}  \}} \\
 & = 
\ccrc{\{ \eq \backslash s_n(f_1 \otimes \cdots \otimes f_n) \mid 
\forall i \, f_i \in 
\msbf{A}_p  \}} \mbox{, where $\eq: A^{(n)} \longrightarrow A^{\otimes n}$,}
\end{align*}
in which $\msbf{A}_p^{\otimes n}$ is a short for $(\msbf{A}_p)^{\otimes n}$.

\item[(ii)]
(generators for the equaliser $\oble{\msbf{A}}{n}$)
\begin{align*}
 (\oble{\msbf{A}}{n})_p &   =  
\ccrc{
\left \{  \eq \backslash s_n( (f_1 \& \iota_1) \otimes \cdots
 \otimes (f_n \& \iota_n )) \mid 
\forall i  \, \, (f_i \in 
\msbf{A}_p  \, \, \mbox{and} \, \,  \iota_i \in \ccrc{\{
 \operatorname{Id}_I\} }) 
\right \} } \\ & 
= 
\ccrc{\{ \eq \backslash  s_n( (f_1 \& \operatorname{Id}_I) \otimes \cdots
 \otimes (f_n \& \operatorname{Id}_I )) \mid 
 \forall i \, f_i \in 
\msbf{A}_p   \} }
\mbox{, where
$\eq: \oble{A}{n} \longrightarrow \obwt{A}{n}$.}
\end{align*}
\end{itemize}
\end{prop}
\begin{proof}
\noindent (i)
The first equality is by the characterisation of the invariant morphisms by the barycenters and the second equality is
by Lemma \ref{gencomcc} (i).

\noindent (ii)
We prove the second line as
the first one is by (i).
First note that $\iota_i \leq \operatorname{Id}_I$
as $\iota_i \in \ccrc{\{
 \operatorname{Id}_I \}}$ is a 
function on $\{* \}=I$ to $\interval{0}{1}$ (while
 $\operatorname{Id}_I$ maps $*$ to 1).
Thus 
$f_i \otimes \iota_i \leq f_i \otimes \operatorname{id}_I$
for all $i=1, \ldots ,n$,  hence 
$s_n( (f_1 \& \iota_1) \otimes \cdots
 \otimes (f_n \& \iota_n ))
\leq 
s_n( (f_1 \& \operatorname{Id}_I) \otimes \cdots
 \otimes (f_n \& \operatorname{Id}_I ))$.
Hence by the downward closed Lemma \ref{gencomcc} (ii), 
$\eq \backslash s_n( (f_1 \& \iota_1) \otimes \cdots
 \otimes (f_n \& \iota_n ))$ belongs to 
RHS of the first equation of (ii), which has shown the assertion.
\end{proof}

\subsection{\normalsize
Distribution of Monoidal Product over Limit in
$\T{\mc{I}}{\opTsKer}$} \label{sect4.2}
This subsection is concerned with showing the distributivity of
the tensor product over the limit L$_{\msbf{A}}$ in
$\T{I}{\cC}$ when $\cC=\opTsKer$.
Using the monotone convergence theorem,
Lebesgue integral over the limit measurable
space is shown to be a convergence of sequence of integrals over the
equalisers along $p_{\infty,n} \otimes \mc{B}$s
(Proposition \ref{approxnun}). For the convergence,
the reciprocity of $p_{\infty, n}$ in Proposition \ref{recippB} above 
plays a crucial role.
The convergence leads to the satisfaction of the distribution
condition (Theorem \ref{conhold}).
The satisfaction is demonstrated by estimating the convergence 
over the barycentoric generators studied above in Section \ref{sect4.3}.


\begin{defn}[sequence of measures $\nu_n$
on $\oble{\mc{A}}{n} \otimes \mc{B}$ for a measure $\nu$ on
$\oble{\mc{A}}{\infty} \otimes \mc{B}$] \label{nun} {\em 
For any measure $\nu: \oble{\mc{A}}{\infty} \otimes \mc{B}
 \longrightarrow \mc{I}$ in $\opTsKer$ and a natural number $n$,
a measure $\nu_n$ on $\oble{\mc{A}}{n} \otimes \mc{B}$ is defined
as follows using
$G_{n, \infty}$ of (\ref{Ginfn});
\begin{align*}
& \nu_n : \oble{\mc{A}}{n} \otimes \mc{B} \longrightarrow \mc{I}
 \quad \quad X \times Y \longmapsto 
\nu(G_{n, \infty}(X) \times Y) 
\end{align*}
for every  
rectangle $X \times Y$ with measurable sets $X \in
\oble{\mc{A}}{n}$ and $Y \in \mc{B}$.

Note the measure $\nu_n \Comp
(p_{\infty, n} \otimes \mc{B})$, which
by Proposition \ref{recippB}, 
is the push forward measure 
$\nu_n \Comp (G_{n, \infty} \times B)^{-1}$
of $\nu_n$ along $G_{n, \infty} \times B$.
Thus $\nu_n \Comp
(p_{\infty, n} \otimes \mc{B})$ consequently 
gives the restriction of the measure $\nu$ 
to the family of measurable subsets
$(\underset{k \leq n}{\bigwith}
\mc{A}^{(k)}) \otimes  \mc{B}$ in 
$\oble{\mc{A}}{\infty}  \otimes \mc{B}$.
}\end{defn}

\bigskip
In terms of $\nu_n$ defined above,
a sequence of measures is constructed to converge to $\nu$.
\begin{prop}[$\{ \nu_n \Comp
(p_{\infty, n} \otimes \mc{B}) \}_{n \in \mathbb{N}}$
converges to $\nu$ when  $n \to
 \infty$] \label{approxnun}
For any measurable function $f: \mc{I} \longrightarrow 
\oble{\mc{A}}{\infty} \otimes \mc{B}$ in $\opTsKer$,
\begin{align}
\int_{\oble{A}{\infty} \times B} f \, \, d \nu 
& = \lim_{n \to \infty}
\int_{\oble{A}{\infty} \times B}
f \, \, d ( \nu_n \Comp (p_{\infty, n} \otimes \mc{B}))
\, \,  \label{unun} \\
& \geq \int_{\oble{A}{\infty} \times B}
f \, \, d ( \nu_n \Comp (p_{\infty, n} \otimes \mc{B}))
\, \,  \nonumber 
\end{align}
The equation (\ref{unun}) stipulates that the measures 
$\nu_n \Comp (p_{\infty, n} \otimes \mc{B})$ converge
to the measure $\nu$ when $n$ tends to infinity.
\end{prop}
\begin{proof}
We prove Equation (\ref{unun}) using the
monotone convergence theorem (as the inequality is direct
for any $n$ by the
definition $\nu_n$). Given $f$,
we define an increasing
sequence $0 \leq f_0 \leq f_1 \leq \cdots \leq f_n \leq  \cdots$
of measurable functions $f_n:
\mc{I} \longrightarrow
\oble{\mc{A}}{\infty} \otimes \mc{B}$ by
\begin{align*}
f_n (((k, a_1 \cdots a_k),b)) :=
\left\{ \begin{array}{lc}
0 & \mbox{for $k > n$} \\
f (((k, a_1 \cdots a_k),b))
 & \mbox{for $k \leq n$}
   \end{array} \right. 
\end{align*}
This yields
\begin{align}
(p_{\infty, n} \otimes \mc{B}) \Comp f
=
(p_{\infty, n} \otimes \mc{B}) \Comp f_n \quad \mbox{for any $n$} \label{ffn}
\end{align}
Obviously from the definition,
\begin{align*}
\lim_{n \to \infty} f_n =f
\end{align*}
The following first equation is 
by the definition $\nu_m$ of Definition \ref{nun},
\begin{align*}
 \int_{\oble{A}{\infty} \times B } f_n \,  d \nu
& = 
 \int_{\oble{A}{\infty} \times B} f_n \,
d ( \nu_n \Comp (p_{\infty, n} \otimes \mc{B}))\\
& 
 = 
\int_{\oble{A}{n} \times B} 
(p_{\infty,n} \otimes \mc{B}) \Comp
f_n \, d \nu_n  \tag*{by Prop \ref{recippB} } \\
&  = 
\int_{\oble{A}{n} \times B} 
(p_{\infty,n} \otimes \mc{B}) \Comp
f \, d \nu_n  \tag*{by (\ref{ffn})} \\
& 
=  \int_{\oble{A}{\infty} \times B} f \,
d ( \nu_n \Comp (p_{\infty, n} \otimes \mc{B}))
\tag*{by Prop \ref{recippB} } 
\end{align*}
Taking the limit $n \to \infty$ yields the assertion with a bypass of
the Lebesgue monotone convergence theorem commuting the limit and the
integral:
\begin{align*}
\int_{\oble{A}{\infty} \times B} f \,  d \nu
& =
\lim_{n \to \infty} \int_{\oble{A}{\infty} \times B}  f_n \, d \nu
 =  
\lim_{n \to \infty} \int_{\oble{A}{\infty} \times B} f \,
d ( \nu_n \Comp (p_{\infty, n} \otimes \mc{B}))
\end{align*}
\end{proof}
Finally with Proposition \ref{approxnun},
we are ready to prove the goal of this subsection.
For the goal, we prepare two technical lemmas, in which
$\bar{f}$ denotes $f \& \operatorname{Id}_I$.

\begin{lem} \label{pnk} 
For any natural numbers $k \leq n$
and measurable functions
$f_i : \mc{I} \longrightarrow \mc{A}$ with $i=1, \ldots ,n$,
the following
equality holds in $\mc{I} \longrightarrow \oble{\mc{A}}{k}$
\begin{align*}
p_{n, k} \Comp \bra{f}{n}=
\frac{(n-k)!}{n!}
\, \eq \backslash \! \! \!    
\sum_{\iota : \setone{k} \hookrightarrow \setone{n}
} 
\bar{f}_{\iota(1)} \otimes \cdots \otimes \bar{f}_{\iota(k)}, 
\end{align*}
where 
$p_{n,k} :=p_{k+1,k} \Comp \cdots \Comp p_{n,n-1}$
and $\iota$ is an inclusion function with $\setone{m}$ denoting the set $\{1, \ldots , m \}$.
\end{lem} 
\begin{proof}
The following starts from $LHS$
and ends with $RHS$ under the instantiating at $(I, x_1 \cdots x_k)$ with
$x_1 \cdots x_k  \in \oble{\mc{A}}{k}$.
\begin{align*}
p_{n.k} \Comp
\bra{f}{n} (I, x_1 \cdots x_k) 
&  = \textstyle \bra{f}{n} (I, x_1 \cdots x_k \overbrace{(2,*) \cdots (2,
 *)}^{n-k}) \tag*{by (\ref{pmn})}\\
& \textstyle = \frac{1}{n!}
\sum_{\sigma \in \mathfrak{S}_n}
\bar{f}_{\sigma(1)}(I, x_1) \cdots \bar{f}_{\sigma(k)}(I, x_k)
\operatorname{Id}_I (I, *) \cdots \operatorname{Id}_I (I, *)  \\
&  \textstyle = \frac{1}{n!}
\sum_{\sigma \in \mathfrak{S}_n} \prod_{i=1}^k \bar{f}_{\sigma(i)} (I, x_i)
=  \frac{(n-k)!}{n!}
\sum_{\iota :  \setone{k}  \hookrightarrow \setone{n}
} \prod_{i=1}^k \bar{f}_{\iota(i)} (I, x_i) \\
& \textstyle = 
\frac{(n-k)!}{n!}
\, \eq \backslash     
\sum_{\iota : \setone{k} \hookrightarrow \setone{n}}
\bar{f}_{\iota(1)} \otimes \cdots \otimes \bar{f}_{\iota(k)}
(I, x_1 \cdots x_k) 
\end{align*}
\end{proof}

\smallskip

\begin{lem} \label{lemdistlim} 
For any measure $\nu$,
\begin{enumerate}
 \item
If $\ort{\nu}{\oble{\mc{A}}{\infty} \otimes \mc{B}
}{\textstyle \bigcup_{n \in \mathbb{N}} \{  \underset{k \in \mathbb{N}}{\bigwith}
(\frac{f_1 + \cdots + f_n}{n})^{(k)} \mid 
	      \, \, \forall i \, f_i \in \msbf{A}_p  \}
\otimes \msbf{B}_p
}$, then 
\begin{align}
\ort{
\frac{m!}{m^k (m-k)!}
\nu_k}{\oble{\mc{A}}{k} \otimes \mc{B}
}{ (p_{m,k} \otimes B) \Comp (\oble{\msbf{A}}{m} \otimes \msbf{B})_p 
} \quad \quad \mbox{for all $m > k$}\label{mknu}
\end{align}
\item
In particular, if $\nu \in \crc{((\oble{\msbf{A}}{\infty} \otimes \msbf{B})_p)}$, then the orthogonality (\ref{mknu}) holds.
\end{enumerate}
\end{lem}
\begin{proof}{}
(2) is direct from (1) as RHS of the premise ortogonality of (1)
is contained in
$(\oble{\msbf{A}}{\infty} \otimes \msbf{B})_p$
(cf. the obvious parts $(ii) \subset (iii) \subset (i)$ in the proof of
Theorem \ref{charlim}). \\
We prove (1).
The premise of (1) is the following first inequality for any $b \in \msbf{B}_p$:
\begin{align*}
 1 & \geq
\inpro{ 
(  \underset{n \in \mathbb{N}}{\with}
(\frac{\bar{f}_1 + \cdots + \bar{f}_m}{m})^{(n)} ) \otimes b}{\oble{\mc{A}}{\infty} \otimes \mc{B}}{ \nu}  \\
&  = 
\lim_{k \to \infty}
\inpro{ 
(  \underset{n \in \mathbb{N}}{\with}
(\frac{\bar{f}_1 + \cdots + \bar{f}_m}{m})^{(n)} ) \otimes b}{
\oble{\mc{A}}{\infty} \otimes \mc{B}}{ \nu_k 
\Comp (p_{\infty, k} \otimes B)}  \tag*{by Prop \ref{approxnun}} \\
& \geq \inpro{ 
(  \underset{n \in \mathbb{N}}{\with}
(\frac{\bar{f}_1 + \cdots + \bar{f}_m}{m})^{(n)} ) \otimes b}{\oble{\mc{A}}{\infty} \otimes \mc{B}}{ \nu_k 
\Comp (p_{\infty, k} \otimes B)}  \tag*{by the def $p_{\infty,k}$} \\
&  = \inpro{ 
(  (p_{\infty, k} \otimes B) \Comp \underset{n \in \mathbb{N}}{\with}
(\frac{\bar{f}_1 + \cdots + \bar{f}_m}{m})^{(n)} ) \otimes b}{
\oble{\mc{A}}{k} \otimes \mc{B}}{\nu_k} \tag*{by reciprocity}  \\
&  = \inpro{ 
(\frac{\bar{f}_1 + \cdots + \bar{f}_m}{m})^{(k)}  \otimes b}{
\oble{\mc{A}}{k} \otimes \mc{B}}{ \nu_k}
\\
& \geq  
\frac{1}{m^k}
\inpro{(\eq \backslash \! \! \!     \sum_{\iota : \setone{k} \hookrightarrow \setone{m}} 
\bar{f}_{\iota(1)} \otimes \cdots \otimes \bar{f}_{\iota(k)} ) \otimes b 
}{\oble{\mc{A}}{k} \otimes \mc{B}}{\nu_k}
\tag*{by the pointwise order between measurable functions  
$(\bar{f}_1 + \cdots + \bar{f}_m)^{\otimes k} \geq \sum_{i : \setone{k}
\hookrightarrow \setone{m}} \bar{f}_{\iota(1)} \otimes \cdots \otimes
 \bar{f}_{\iota(k)}$} \\
& \geq  
\frac{1}{m^k} \frac{ m!}{(m-k)!}
\inpro{
(p_{m, k} \otimes B) \Comp (\bra{f}{m} \otimes b)
}{\oble{\mc{A}}{k} \otimes \mc{B}}{\nu_k}
\tag*{by Lemma \ref{pnk}}
\end{align*}
The above means  
\begin{align*}
\ort{
\frac{m!}{m^k (m-k)!}
\, \nu_k}{}{ (p_{m,k} \otimes B) \Comp (G \otimes \msbf{B}_p)
}, \end{align*}
where $G$ denotes the scope of the double orthogonal of
(ii) in Proposition \ref{corgene} so that
$\ccrc{G}=(\oble{\msbf{A}}{m})_p$.
\\ This implies by reciprocity and $\crc{X}=\cccrc{X}$
(i.e., $\ort{r}{}{X}$ iff $\ort{r}{}{\ccrc{X}}$),  
\begin{align*}
& \ort{
\frac{m!}{m^k (m-k)!}
\nu_k}{}{ (p_{m,k} \otimes B) \Comp \ccrc{(G \otimes \msbf{B}_p)}
} 
\end{align*}
This is the assertion as 
\begin{align*}
\ccrc{(G \otimes \msbf{B}_p)}=
\ccrc{(G \otimes \ccrc{(\msbf{B}_p)})}=
\ccrc{(\ccrc{G} \otimes \msbf{B}_p )}=
\ccrc{((\oble{\msbf{A}}{m})_p \otimes \msbf{B}_p )}=
(\oble{\msbf{A}}{m} \otimes \msbf{B})_p,
\end{align*}
whose second equation is by the stable tensor of Lemma \ref{intstabten}. 
\end{proof}

\begin{thm}[$\otimes$ distributes over L$_{\msbf{A}}$
in $\T{\mc{I}}{\opTsKer}$] \label{conhold}
The condition (\ref{condition}) in Proposition \ref{condist}
holds for all objects $\msbf{A}$ and $\msbf{B}$ in $\T{\mc{I}}{\opTsKer}$.
That is, the distributivity of the monoidal product over L$_{\msbf{A}}$
is retained in the orthogonal category.
 \end{thm}

\begin{proof}
 We shall show the following (\ref{altcondition}) for any $u$,
 equivalent to (\ref{condition});
\begin{align*}
\ort{
\forall n \, (p_{\infty,n} \otimes B) \Comp u}{}{
\crc{((\oble{\msbf{A}}{n})_p \otimes \msbf{B}_p)} }
\Longrightarrow 
\ort{u}{}{
\crc{((\oble{\msbf{A}}{\infty})_p \otimes \msbf{B}_p)} }
\end{align*}
Take any $\nu \in \crc{((\oble{\msbf{A}}{\infty})_p \otimes \msbf{B}_p)}$
in order to show $\ort{u}{}{\nu}$.  The premise of the assertion is 
for all $m$; 
$$ u_m:= (p_{\infty,m} \otimes B) \Comp u \in
\ccrc{((\oble{\msbf{A}}{m})_p \otimes \msbf{B}_p)}
= (\oble{\msbf{A}}{m} \otimes \msbf{B})_p$$
Thus the premise implies by the second assertion of Lemma \ref{lemdistlim} 
\begin{align*}
\frac{m^k (m-k) !}{m!} \geq 
\int_{\oble{A}{k} \times B}  ((p_{m, k} \otimes B) \Comp u_m) \, d \nu_k 
\end{align*}
Taking the limit on $m$,
\begin{align}
1= \lim_{m \to \infty} \frac{m^k (m-k) !}{m!} \geq 
\lim_{m \to \infty} \int_{\oble{A}{k} \times B}  ((p_{m, k} \otimes B) \Comp u_m) \, d \nu_k =
\int_{\oble{A}{k} \times B}  ((p_{\infty, k} \otimes B) \Comp u) \, d \nu_k
\label{firstlimit}
\end{align}
The last equality is Lebesgue monotone convergence and the first equality is by
\begin{align}
\textstyle \lim_{m \to \infty} \frac{(m-k)! n^k}{m!}=
\lim_{m \to \infty} \frac{m^k}{m(m-1) \cdots (m-(k-1))}  =
\textstyle \lim_{m \to \infty} \frac{1}{1 (1 - \frac{1}{m}) \cdots (1 -
 \frac{k-1}{m})}=1 \label{mkconv}
\end{align}
Thus taking the limit of (\ref{firstlimit}) now on $k$,
\begin{align*}
1 \geq \lim_{k \to \infty}
\int_{\oble{A}{k} \times B}  ((p_{\infty, k} \otimes B) \Comp u) \, d \nu_k
= 
\lim_{k \to \infty}
\int_{\oble{A}{\infty} \times B}  u \, d \nu_k \Comp (p_{\infty, k} \otimes B) 
= \int_{\oble{A}{\infty} \times B} u \, d \nu
\end{align*}
\end{proof}

By Theorem \ref{conhold} and Proposition \ref{PROPconditiontnteneq},
we finally have  
\begin{cor} \label{corconhold}
$\T{\mc{I}}{\opTsKer}$
has the free exponential whose forgetful image by $\T{\mc{I}}{\opTsKer}
\longrightarrow \opTsKer$ is the free exponential of Theorem \ref{ftsect3}.
\end{cor} 
\begin{proof}
By Theorem \ref{mainsec1} whose conditions (\ref{condition}) and
(\ref{conditiontnteneq}) hold respectively by  
Theorem \ref{conhold} and Proposition \ref{PROPconditiontnteneq}.
\end{proof}

\subsection{
\normalsize Characterising Limit $\oble{\msbf{A}}{\infty}$ in
$\T{\mc{I}}{\opTsKer}$} \label{sect4.5}
This subsection is concerned with a concrete representation of
the free exponential of 
$\T{\mc{I}}{\opTsKer}$ guaranteed abstractly in Corollary
\ref{corconhold} above.
Using the characterisation of the equalisers in 
Section \ref{sect4.3} globally over natural numbers $n$,
a characterisation
of the limit $\oble{\msbf{A}}{\infty}$ 
is obtained in Theorem \ref{charlim} in direct terms of
$\msbf{A}_p$ within $\T{\mc{I}}{\opTsKer}$. 
The limit characterisation directly leads to the coincidence of the free
exponential with the exponential of
the linear comonad in our preceding study \cite{HamLEC}. 

Applying Proposition \ref{proplimT} to $\cC=\opTsKer$,
whereby the set $(\oble{\msbf{A}}{\infty})_p$
is specified using  the description (\ref{factotau}) 
in Theorem \ref{limTK} of 
the mediating morphism $x_\infty$ for $\opTsKer$, we have;
\begin{align}
(\oble{\msbf{A}}{\infty})_p =
\ccrc{ \left \{
\begin{aligned} 
 \, \,  x_\infty = \textstyle \underset{k \in \mathbb{N}}{\bigwith} 
\, (\pr_l^{(k)} \Comp x_k) \, \, 
\mid \, \{ x_n: \msbf{I} \longrightarrow \oble{\msbf{A}}{n} \}_n 
\, \, \mbox{is } \\
\, \, \mbox{a cone to the diagram $\{ p_{n+1,n} \}_n$ in 
$\T{\mc{I}}{\opTsKer}$} 
\end{aligned}
\right \} } \label{xinf}
\end{align}

\begin{thm}[characterising the limit $\oble{\msbf{A}}{\infty}$
in $\T{\mc{I}}{\opTsKer}$] \label{charlim}
The following three subsets of $\opTsKer(\mc{I}, \oble{\mc{A}}{\infty})$
all coincide:
\begin{enumerate}
 \item[(i)] $(\oble{\msbf{A}}{\infty})_p$
\item[(ii)]  
$\begin{aligned}
\textstyle \ccrc{\left(
\bigcup_{n \in \mathbb{N}} \{  \underset{k \in \mathbb{N}}{\bigwith}
(\frac{f_1 + \cdots + f_n}{n})^{(k)} \mid 
	      \, \, \forall i \, f_i \in \msbf{A}_p  \}
\right)}
\end{aligned}$
\item[(iii)]
$\begin{aligned}
\ccrc{\{ \, \,  \textstyle
 \underset{k \in \mathbb{N}}{\bigwith} g^{(k)}
 \, \, \mid 
g \in 
\msbf{A}_p  \} }
\end{aligned}$
\end{enumerate}
\end{thm}

\begin{proof} Circular inclusions 
$(i) \subset (ii) \subset (iii) \subset (i)$
are shown in the proof.

\noindent ($(ii) \subset (iii)$) Obvious: For all $n$, 
$\frac{f_1 + \cdots + f_n}{n} \in \msbf{A}_p$ as
 $\forall i \,  \inpro{f_i}{}{\nu}
\leq 1 \Rightarrow \inpro{\frac{f_1 + \cdots + f_n}{n}}{}{\nu} \leq 1$.

\noindent ($(iii) \subset (i)$) 
Straightforward because the generators for (i)
contain those for (iii). That is, $\{ \oble{g}{n}:
\msbf{I} \longrightarrow \oble{\msbf{A}}{n}\}_n$ forms a cone
whose mediating morphism belonging to (\ref{xinf})
is described by the specific form as follows:
\begin{align*}
\textstyle \underset{k \in \mathbb{N}}{\mathlarger{\&}} (\pr_l^{(k)} \Comp \oble{g}{k}) & = \textstyle
\underset{k \in \mathbb{N}}{\mathlarger{\&}} (\pr_l^{(k)} \Comp (g \& I)^{(k)}) 
 \textstyle = \underset{k \in \mathbb{N}}{\mathlarger{\&}} (\pr_l \Comp (g \& I))^{(k)} =
\underset{k \in \mathbb{N}}{\mathlarger{\&}} g^{(k)}. 
\end{align*}

\noindent ($(i) \subset (ii)$) Since both $(i)$ and $(ii)$ are double
 orthogonal homsets in $\opTsKer$, it suffices to prove
\begin{align}
\textstyle
\crc{ (
(\oble{\msbf{A}}{\infty})_p)}
\supset 
\crc{ \left(
\bigcup_{n \in \mathbb{N}} \{ \underset{k \in \mathbb{N}}{\bigwith} 
(\frac{f_1 + \cdots + f_n}{n})^{(k)} \mid 
	      \, \, \forall i \, f_i \in \msbf{A}_p  \}
\right)} \label{1impl2}
\end{align}
Take an arbitrary measure $\nu$ from RHS of (\ref{1impl2}).
Take any generator $x_\infty \in (\oble{\msbf{A}}{\infty})_p$
of (\ref{liminf}) of Lemma \ref{proplimT} 
in order to show $\ort{\nu}{}{x_\infty}$. By Proposition  \ref{approxnun} 
\begin{align*}
\inpro{x_\infty}{\oble{A}{\infty}}{\nu}
= \lim_{k \to \infty}
\inpro{x_\infty}{\oble{A}{\infty}}{\nu_k \Comp p_{\infty, k}}
\stackrel{reciproc}{=}
\lim_{k \to \infty}
\inpro{p_{\infty, k} \Comp x_\infty}{\oble{A}{m}}{\nu_k} 
\end{align*}
On the other hand by the second assertion of Lemma \ref{lemdistlim} (with $\mathbb{B}=
\msbf{I}$)  
\begin{align*}
& \ort{ 
\frac{m!}{m^k (m-k)!}
\nu_k}{}{ p_{\infty, k} \Comp x_\infty} \in  p_{m, k} \Comp
p_{\infty, m} \Comp  (\oble{\msbf{A}}{\infty})_p \subset
p_{m, k} \Comp (\oble{\msbf{A}}{m})_p
\end{align*}
This yields the following inequality, followed by the equality
(\ref{mkconv}).
\begin{align*}
\lim_{k \to \infty}
\inpro{p_{\infty, k} \Comp x_\infty}{\oble{A}{m}}{\nu_k} 
\leq \lim_{k \to \infty} \frac{m^k (m-k)!}{m!} = 1 
\end{align*}

\end{proof}

\begin{rem}[On Theorem \ref{charlim}] \label{instcharlim}
{\em Given that each generator of (iii)
is a measurable function on the measurable space $\oble{\mc{A}}{\infty}
=\underset{n \in \mathbb{N}}{\bigwith} \mc{A}^{(n)}$,
its instantiation at each point $(n, x_1 \cdots x_n)$
from the $n$-th componential space with $x_1 \cdots x_n \in \mc{A}^{(n)}$ is explicitly calculated as follows:
\begin{align*}
    (\textstyle \underset{k \in \mathbb{N}}{\bigwith} 
\, g^{(k)})(I, (n, x_1 \cdots x_n)) &   = 
g^{(n)}(I, x_1 \cdots x_n) \\
& 
=  
g^{\otimes n}(F^{-1}(I^{(n)}), (x_1, \ldots, x_n)) \\
& 
= 
g^{\otimes n}(I^{\otimes n}, (x_1, \ldots , x_n)) \tag*{
by $F^{-1}(I^{(n)}) = I^{\otimes n}$}\\
& 
= \textstyle \prod_{i=1}^n g(I, x_i)
\end{align*}
}\end{rem}

\smallskip

The exponential constructed in the present paper coincides with
that for linear exponential command independently studied in our preceding 
\cite{HamLEC} (free from  Melli\`{e}s-Tabareau-Tasson construction).
\begin{cor} \label{corinstcharlim}
The free exponential of $\T{\mc{I}}{\opTsKer}$
coincides with the exponential structure 
for the linear exponential comonad of the tight (double glueing)
category ${\bf T}(\opTsKer)$ in \cite{HamLEC}. 
\end{cor}
\begin{proof}
The instantiation of Remark \ref{instcharlim} coincides with 
the natural transformation $\sf{k} : \opTsKer
(\mc{I}, (-))  \longrightarrow \opTsKer
 (\mc{I}, (-)_e)$ defined in Definition 4.4 of \cite{HamLEC} making $\opTsKer (\mc{I}, (-))$ linear
 distributive, where $(~)_e$ is the exponential action directly formulated in
 \cite{HamLEC}  using the exponential measurable
 spaces in \cite{HamLEC}.
See Section 4.1 of \cite{HamLEC} how the natural transformation
 $\sf{k}$ (Definition 4.4 of \cite{HamLEC}) yields the exponential structure
of the tight double glueing.
\end{proof}

\begin{rem}[Comparison to Crubill\'{e} et al 
\cite{Crubille}'s characterisation of
the limit for $\Pcoh$]{\em
As seen in Corollary \ref{corinstcharlim},
our limit characterisation in this subsection directly 
coincides with the linear exponential comonad in our preceding \cite{HamLEC}.
In contrast in \cite{Crubille} for probabilistic coherent
spaces (shown to be a discretisation of 
$\T{\mc{I}}{\opTsKer}$ in Section \ref{sect4.4} below),
Crubill\'{e} et al first employed a different
characterisation (from our (ii) and (iii) of
Theorem \ref{charlim}) of the free exponential, and 
then separately showed its equivalence to the original exponential of
\cite{DE}.
Accordingly in $\T{\mc{I}}{\opTsKer}$,
our methodology of Theorem \ref{charlim} can accommodate 
their characterisation
$\dbra{f}{n}$ giving an inverse image of
$p_{n,m} \Comp \eq \backslash s_n (\bar{f}_1 \otimes \cdots \otimes
 \bar{f}_n)$ under $p_{\infty,m}$. 
See Appendix \ref{accodbra} how to accommodate
their characterisation into our Theorem \ref{charlim}.}\end{rem}

\subsection{\normalsize Discretisation} \label{sect4.4}
The purpose of this final subsection is to show that the category of probabilistic coherent spaces arises from $\T{\mc{I}}{\opTsKer}$ by discretisation.
In order to see this precisely, we start with defining a continuous subcategory.  The definition simply employs a generalisation of
the technical conditions (non-zero and 
boundedness) imposed on each object of
probabilistic coherence spaces \cite{DE}.

\begin{defn}[The subcategory $\Sub$ of positively non-zero bounded objects]{\em 
The subcategory $\Sub$ of $\T{\mc{I}}{\opTsKer}$
consists of the objects 
$\msbf{A}=(\mc{A}, \msbf{A}_p)$ such that (\ref{nonzeroboun}) holds for all $a \in A$,
where $A$ is the underlying set of $\mc{A}$.
\begin{eqnarray}
0 < \sup \msbf{A}_p(a) := \sup  \{ f(a) \mid f \in \msbf{A}_p \} < \infty
\label{nonzeroboun}   \end{eqnarray}
Note each $f \in \msbf{A}_p$ is an element of $
\opTsKer (\mc{I}, \mc{A})$,
hence a measurable function on $\mc{A}$. \\
In what follows, for an object $\msbf{A}$ of the subcategory
and $a \in A$, 
$\C{\msbf{A}}{a}$ denotes $\sup \msbf{A}_p(a)$ belonging to $(0, \infty)$. 
}\end{defn}

\begin{prop} \label{submonofree}
\begin{itemize} 
 \item[(i)]
$\Sub$ is a monoidal subcategory with cartesian product of
$(\T{\mc{I}}{\opTsKer}, \otimes, \msbf{I})$ with the $\&$.
\item[(ii)]
The free exponential for $\T{\mc{I}}{\opTsKer}$
of Corollary \ref{corconhold}
becomes by the restriction that of the subcategory $\Sub$.
\end{itemize}

\end{prop}
\begin{proof}
(i) The both products are shown to be preserved
inside the subcategory.

\noindent (monoidal product $\otimes$) \\
First, $\msbf{I}$ is an object of the subcategory as any element $\crc{\{ \operatorname{Id}_I \}}$ is subMarkov in an obvious sense. Second, we prepare the following claim:
For any object $\msbf{A}$ in $\Sub$ and $a \in A$,
$\C{\msbf{A}}{a}^{-1} \delta_a \in \crc{(\mathbb{A}_p)}$, where $\delta_a$
is the Dirac delta measure. The clam is valid
as $\forall f \in \mathbb{A}_p$, $\inpro{f}{}{\C{\msbf{A}}{a}^{-1} \delta_a} = \C{\msbf{A}}{a}^{-1} \int f d \delta_a =
 \C{\msbf{A}}{a}^{-1} f(a) \leq 1$.

We show that 
\begin{eqnarray*}
0 \not = \C{A \otimes B}{(a,b)} \leq \C{\msbf{A}}{a} \C{\msbf{B}}{b} 
\end{eqnarray*}
(nonzero) Obvious:
Given any $(a,b)$, we can take $f \in \mathbb{A}_p$ and $g \in \mathbb{B}_p$
which are non zero at $a$ and $b$, respectively. then $f \otimes g \in 
\in \mathbb{A}_p \otimes \mathbb{B}_p$ is non zero at $(a,b)$.

\noindent 
(inequality) For any measures $\mu \in \crc{(\mathbb{A}_p)}$ and $\tau \in \crc{(\mathbb{B}_p)}$, Fubini-Tonelli
$\int_{A \times B} (f \otimes g) d(\mu \otimes \tau) = 
(\int_A f d \mu)(  \int_B g d \tau)$ assures 
$\mu \otimes \tau \in \crc{(\mathbb{A}_p \otimes\mathbb{B}_p) }$. Thus in particular taking
Dirac measures $\delta_a$ and $\delta_b$ divided by the scalars in the claim,
\begin{eqnarray*}
\C{\msbf{A}}{a}^{-1} \C{\msbf{B}}{b}^{-1} \delta_{(a,b)} \in \crc{(\mathbb{A}_p \otimes\mathbb{B}_p)},
 \end{eqnarray*}
which implies for all $h \in \ccrc{(\mathbb{A}_p \otimes\mathbb{B}_p )}$,
$h(a,b)= \textstyle \int_{A \times B} h \, d \delta_{(a,b)} \leq  \C{\msbf{A}}{a} \C{\msbf{B}}{b}$.
 
\smallskip
\noindent (cartesian product) Obvious as
$\C{\&_{\rm i \in I} A_i}{(i,a)}=
\C{A_i}{a}$ for each $i \in I$ and $a \in A_i$.

\bigskip

\noindent (ii)
In $\opTsKer$, any
$\msbf{a} \in \oble{\mc{A}}{\infty}= 
\underset{k \in \mathbb{N}}{\bigwith} \mc{A}^{(k)}$
is of the form $(n, a_1 \cdots a_n)$ with 
$a_1 \cdots a_n \in \mc{A}^{(n)}$ for certain $n$.
Then it suffices to show the following
as the RHS is finite from (i) by the definition of $\oble{\msbf{A}}{n}$. 
\begin{eqnarray*}
0 \not = \C{\oble{\msbf{A}}{\infty}}{(n, a_1 \cdots a_n)} 
\leq \C{\oble{\msbf{A}}{n}}{a_1 \cdots a_n} 
\end{eqnarray*}

\noindent (inequality)
By Proposition \ref{proplimT}, for any $x_\infty \in
(\oble{\msbf{A}}{\infty})_p$,  $x_\infty ((n, a_1 \cdots a_n))
= (p_{\infty, n} \Comp x_\infty) (a_1 \cdots a_n) =
x_n(a_1 \cdots a_n)$ with $x_n \in 
(\oble{\msbf{A}}{n})_p$. 
As $p_{\infty,n} \Comp (\oble{\msbf{A}}{\infty})_p
\subset (\oble{\msbf{A}}{n})_p$, the inequality is derived.

\smallskip

\noindent (nonzero) Similar as that for (i).
For any given $\msbf{a}=(n, a_1 \cdots a_n) \in \oble{A}{\infty}$,
we can take a function $f_i \in \msbf{A}_p$ whose value at $a_i$ is non zero.
Take $x_\infty = \underset{k \in \mathbb{N}}{\bigwith} g^{(k)}$ with 
 $g: =\frac{f_1 + \cdots + f_n}{n}$, which is an element of
$(\oble{\msbf{A}}{\infty})_p$.
Then $x_\infty (\msbf{a})=
\frac{1}{n^n}
\prod_{i} \sum_{j} f_j (a_i) \not = 0$ as each $f_i$ is $\zeroinf$-valued.
\end{proof}

\bigskip

Finally let us go to discretisation.

\begin{defn}{\em
 The discretisation $\TsKeromg$ of $\TsKer$
is the full subcategory of the countable measurable spaces
whose $\sigma$-fields are generated by the singleton subsets.
In the discretisation, the composition
in terms of the convolution (\ref{compTKer})
collapses to the products of matrices.
\begin{align}
\iota \Comp \kappa (x,C) = &
\sum_{y \in Y} \kappa(x, \{ y \}) \iota(y,C) \label{compTKeromg}
\end{align}
Obviously $\TsKeromg$ is closed under $\otimes$ and the product.}
\end{defn}

\begin{defn}[$\Pcoh, \otimes, \&$]~\\{\em 
\noindent (inner product)
$\inprocoh{x}{x'} := \sum_{a \in A} x_a x'_a$, for $x, x' \subseteq
\mathbb{R}_{+}^{A}$ with a countable set $A$.

\noindent (polar)
$P^\bot \! := \{ x' \in \mathbb{R}_+^A \mid \forall x \in P \, 
\inprocoh{x}{x'} \leq 1 \}$ for $P \subseteq \mathbb{R}_+^A$. 

\smallskip
The category $\Pcoh$ of the probabilistic coherent spaces is
 defined using the inner product and the polar:

\noindent (object)
$X=(\abs{\! X \!},\Po{X})$, 
where $\abs{\! X \!}$ is a countable set, $\Po{X} \subseteq
 \mathbb{R}_+^{\abs{\, X \,}}$ such that $\Po{X}^{\bot \bot} \subseteq
 \Po{X}$, and $0 <  \sup \{ x_a  \mid  x \in \Po{X} \}
< \infty$ for all $a \in \abs{\! X \!}$.

\smallskip

\noindent (morphism)
A morphism from $X$ to $Y$ is an element $u \in \Po{(X \otimes
Y^\bot)^\bot}$,
which can be seen as a matrix $(u)_{a \in \abs{X}, b \in \abs{Y}}$
of columns from $\abs{\! X \!}$ and of rows from $\abs{\! Y \!}$.
Composition is the product of two matrices such that
$(uv)_{a, c}= \sum_{b \in \abs{Y}} u_{a,b} v_{b,c}$
for $u: X \longrightarrow Y$ and $v: Y \longrightarrow Z$.

\smallskip

\noindent $\Pcoh$ has a monoidal product $\otimes$ and a cartesian product $\&$:

\noindent(monoidal product $\otimes$) \\
$X \otimes Y = (\abs{\! X \!} \times \abs{\! Y \!}, \{ x \otimes y \mid x \in
   \Po{X} \, \,  y \in \Po{Y}   \}^{\bot \bot})$. \\
For $u \in \Pcoh(X_1,Y_1)$ and $v \in \Pcoh(X_2,Y_2)$,
$u \otimes v \in \Pcoh(X_1 \otimes X_2,Y_1 \otimes Y_2)$ is 
$(u \otimes v)_{(a_1,a_2), (b_1,b_2)}= u_{a_1,b_1} v_{a_2,b_2}$.

\smallskip

\noindent (product $\&$)\\
$X_1 \& X_2 = (\abs{\! X_1 \!} \biguplus \abs{\! X_2 \!}, \{ x \in \mathbb{R}_+^{\biguplus_i
   \abs{X_i}} \mid  \forall i \, \, \pi_i(x) \in \Po{(X_i)}  \} )$, \\
where $\pi_i(x)_a$ is $x_{(i,a)}$.
$\Po{(X_1 \& X_2)}$ becomes automatically closed under the bipolar.

}\end{defn}

\begin{prop}
\label{propdisciso}
The two monoidal categories $\Subomg$ and 
$\Pcoh$ with cartesian products
are isomorphic.
\end{prop}
\begin{proof}
First, the positivly bounded and non-zero conditions on $\Pcoh$ objects is exactly the condition (\ref{nonzeroboun}) imposed for  $\Subomg$.
In the discretisation the symmetry arises 
$\operatorname{Hom}(\mc{I}, \mc{X})
=\operatorname{Hom}(\mc{X},\mc{I})$,
which means that measures and measurable
functions become indistinguishable.
Thus the both inner products are the same,
hence so are the objects in the both categories
as the double orthogonality and the bipolar define the same notion.
  The morphisms of the two categories are identical
as they are the matrices and their products.
The definitions of monoidal and (finite) products in the both categories
are directly observed to be the same.
\end{proof}

Moreover, the free exponential construction of $\opTsKer$
in Section \ref{sect3} and its lifting in Section \ref{sect4} 
are a two-layered continuous generalisation of
Crubill\'{e} et al's free exponential for $\Pcoh$ in \cite{Crubille}:
To be precise, directly from 
Proposition \ref{submonofree} (ii);
\begin{cor} \label{cordisciso}
The exponentials of $\Subomg$
and of $\Pcoh$ are isomorphic.
\end{cor}
\begin{proof}
By the universality of the free exponential,
the free exponential of the former category
is isomorphic to that of the latter constructed in \cite{Crubille}.
\end{proof}
We end this section with making 
the isomorphism of Corollary \ref{cordisciso} explicit:
The following is the exponential of $\T{\mc{I}}{\opTsKeromg}$. 

\noindent (exponential on objects) \\
$! \, \msbf{X} = (! \, \mc{X} = \! 
\underset{k \in \mathbb{N}}{\mathlarger{\&}} \mc{X}^{(k)}
 , \, (\oble{\msbf{X}}{\infty})_p)$.
Every generator 
$\textstyle \underset{k \in \mathbb{N}}{\bigwith} 
\, g^{(k)} \in
(\oble{\msbf{X}}{\infty})_p$ is by Remark \ref{instcharlim},
\begin{align*}
(\textstyle \underset{k \in \mathbb{N}}{\bigwith} 
\, g^{(k)} )
 (I, (n, x_1 \cdots x_n)) 
= \prod_{i=1}^n f(I, x_i) 
\end{align*}

\noindent (exponential on morphisms)
For any $\kappa: \msbf{Y} \longrightarrow \msbf{X}$
(hence $\kappa : \mc{Y} \longrightarrow \mc{X}$), 
every $! \, \kappa: ! \, \msbf{Y} \longrightarrow ! \, \msbf{X}$ is
by Proposition \ref{actnlim},
\begin{align*}
 \oble{\kappa}{\infty} (\{ y_1 \cdots y_n \}, x_1 \cdots x_n) 
& = \kappa^{(n)} (\{ y_1 \cdots y_n \}, x_1 \cdots x_n) \\
& = \kappa^{\otimes n} ( \FI(\{ (y_1 \cdots y_n) \}), (x_1, \ldots ,x_n)) 
\textstyle  \\
& \textstyle = \kappa^{\otimes n} ( \biguplus_{\sigma \in
 \mathfrak{S}_n
 / S_{\vec{\msbf{x}}} } \{ (y_{\sigma(1)}, \ldots, y_{\sigma(n)} ) \},
 (x_1, \ldots ,x_n) ) \\
& \textstyle = \sum _{\sigma \in \mathfrak{S}_n / S_{\vec{\msbf{x}}}
} \kappa^{\otimes n} ( \{ (y_{\sigma(1)}, \ldots, y_{\sigma(n)}) \},
 (x_1, \ldots ,x_n)) \\
& \textstyle = \sum _{\sigma \in \mathfrak{S}_n / S_{\vec{\msbf{x}}}
} \prod_{i=1}^n 
\kappa( y_{\sigma(i)}, x_i)
\end{align*}
where $S_{\vec{\msbf{x}}}:=\{ \sigma \in \mathfrak{S}_n \mid 
(x_1, \ldots , x_n) = (x_{\sigma(1)}, \ldots , x_{\sigma(n)})
\}$ is the stabiliser subgroup fixing $\vec{\msbf{x}}=(x_1, \ldots ,
x_n)$. Note that the actions $\sigma$ ranging over
$\mathfrak{S}_n$ modulo $S_{\vec{\msbf{x}}}$
are independent of the ordering 
$\vec{\msbf{x}}$ of $\msbf{x}$.

\smallskip

It is direct to see that the above exponential
action coincides with that in $\Pcoh$
(see \ref{exppcoh}) via the following isomorphisms:
For any set $X$, 
each multiset of size $n$ maps to the unique element
of $X^{(n)}$, which gives the isomorphisms:
\begin{align}
\mathcal{M}_{\rm fin}(X)  \cong
\underset{k \in \mathbb{N}}{\mathlarger{\&}} X^{(n)} 
&  \quad \quad [x_1, \ldots , x_n] \mapsto 
(n, x_1 \cdots x_n) \label{ordirmult}
\end{align}

\smallskip 

\noindent{\bf Future Directions} \\
We remain it a future work 
how to construct a certain monoidal closed structure inside $\TsKer$,
in comparison with the recent development of the higher order
probabilistic programming \cite{HKSY}, in particular with 
\cite{EPT,EPTLics} modelling probabilistic PCF.
We will pursue to relate our transition
kernels to the continuous denotational semantics
of measurable cones and measurable stable functions \cite{EPTpopl},
whose cartesian closed structure
is shown in \cite{CruLICS} to subsume that of $\Pcoh$.
The recent development of integral structure on measurable cones
in \cite{EhrGeo} is significant for this direction, where the closed structure is obtained by the categorical adjoint functor theorem.
From a different perspective, we intend to accommodate 
some probabilistic feed back (or iteration) in the monoidal
structure of s-finiteness, as addressed in \cite{Stat} and analysed
in GoI semantics for Bayesian programming \cite{DLH}.

\section*{Acknowledgment}
The author is deeply grateful to the anonymous reviewer for his 
eagle-eyed reading proofs in many rounds, which has prevented the paper from including specific technical defects.

\section{Appendix}
\renewcommand{\thesection}{A}

\subsection{\normalsize Omitted Proofs}

\subsubsection{Proof of Proposition \ref{protjc}} \label{approtjc}
\begin{proof}
First, to check the double orthogonality, it suffices to show
that $\msbf{A}_p \& \msbf{B}_p = \crc{U}$
for certain subset $U \subseteq \cC(A \& B, J)$.
The $U$ is shown to be given by $\crc{(\msbf{A}_p)} \!  \cdot \pr_l
\cup \crc{(\msbf{B}_p)}  \! \cdot \pr_r$.
In the following the reciprocity of $\bot$
is used wrt the left and right projections $\pr_l$ and $\pr_r$, respectively.

\noindent $(\subseteq)$
Take any $u \& v$ from (LHS) with $u \in \msbf{A}_p$ and  $v \in \msbf{B}_p$.
Then $\forall a \in  \crc{\msbf{A}_p} \, \big( 
\ort{\pr_l  (u \& v) = u }{A}{a} \Leftrightarrow  \ort{u \& v}{A \& B}{a \cdot \pr_l} \big)$
and $\forall b \in  \crc{\msbf{B}_p}
\, \big( 
 \ort{\pr_r  \cdot (u \& v) = u }{B}{b}
\Leftrightarrow \ort{u \& v}{A \& B}{b \cdot \pr_r} \big)$.

\noindent $(\supseteq)$
Take any $\psi$ from (RHS).
Then $\forall a \in \msbf{A}_{cp} \, 
\big( \ort{\psi}{A \& B}{a \cdot  \pr_l}
\Leftrightarrow
 \ort{\pr_l \cdot \psi}{A}{a} \big)$,
and 
$\forall b \in \msbf{B}_{cp} \, 
\big(\ort{\psi}{A \& B}{b \cdot  \pr_r}
\Leftrightarrow 
\ort{\pr_r \cdot \psi}{A}{b} \big)$.
Thus $\psi=\pr_l \cdot \psi \,  \& \,  \pr_r \cdot \psi$
belongs to (LHS). 

\smallskip
\noindent Second, it is direct that
the mediating morphisms for the product resides in $\T{J}{\cC}$:
Given $f : \msbf{C} \longrightarrow \msbf{A}$
and $g : \msbf{C} \longrightarrow \msbf{B}$,
for any $x \in \msbf{C}_p$
$(f \& g) \cdot x = f \cdot x \, \, \& \, \,  g \cdot x \in \msbf{A}_p \&
\msbf{B}_p$.
\end{proof}

\subsubsection{ Demonstration of Eqn.(\ref{square})
in Def \ref{actmor}} \label{apsquare}
\begin{proof}
By definition,
the left and the right squares and the outer most rectangle (with
the two bent horizontal arrows) commute.
This implies that any composition in the diagram
 starting from $\oble{A}{n+1}$
and ending at $\obwt{B}{n}$ defines the same map
equalising the $n!$-symmetries $\obwt{B}{n}$.
Thus there exists the unique morphism from $\oble{A}{n+1}$
to $\oble{B}{n}$ factoring the same map.
The uniqueness implies the commutativity of the middle square,
which is the assertion.
$$
\xymatrix
@C=12pt
{
\obwt{B}{n} &    \oble{B}{n} \ar[l]_{\eq} &  
\ar[l]_{p_{n+1,n}}
\oble{B}{n+1}  \ar[r]^{\eq}
& 
\obwt{B}{n+1} 
\ar@/_20pt/[lll]_{ \obwt{B}{n} \otimes \, \pr_r} \\
\obwt{A}{n} \ar[u]^{(f \cdot \pr_l \& I \cdot \pr_r)^{\otimes n}}
 &    \oble{A}{n} \ar[l]_{\eq} \ar[u]^{\oble{f}{n}} &  
\ar[l]_{p_{n+1,n}}
\oble{A}{n+1}  \ar[r]^{\eq} \ar[u]_{\oble{f}{n+1}}
\ar@{-->}[ul]_{\exists !}
& 
\obwt{A}{n+1} \ar[u]_{(f \cdot \pr_l \& I \cdot \pr_r)^{\otimes (n+1)}} 
\ar@/^20pt/[lll]^{ \obwt{A}{n} \otimes \, \pr_r}
}$$
\end{proof}

\subsubsection{Proof of Lemma \ref{ccrclem}} \label{apccrclem}
\begin{proof}{}
For any $h \in \ccrc{U}$, we shall show that $f \Comp h \in
 \ccrc{f(U)}$: For the assertion,
take any $v \in \crc{f(U)}$. Then for all $u \in U$,
$\ort{v}{B}{f \Comp u}$
iff
$\ort{v \Comp f}{B}{u}$ by the reciprocity.
This means 
$v \Comp f \in \crc{U}$,
hence
$\ort{h}{A}{v \Comp f}$.
Thus by the reciprocity
$\ort{f \Comp h}{B}{v}$,
which is the assertion.
\end{proof}

\subsubsection{Proofs of Theorem \ref{tencMS}} \label{proofdist}
All the following proofs generalise those of
Proposition \ref{eqTK}, Claims 1 of Theorem \ref{limTK}, respectively
in order to accommodate the tensor factor $\otimes \mc{Z}$
consistently using the double integration by Fubini-Tonelli.

\smallskip

\noindent{(Proof of (i))} \\
For any transition kernel $\tau$ of the
codomain $\mc{X}^{(n)} \otimes \mc{Z}$:
\begin{align*}
 (\eq_{\mc{X}} \otimes \mc{Z}) \Comp \tau (-, ((x_1, \ldots, x_n), z))
& =\textstyle  \int_{X^{(n)} \otimes Z} \tau (-, (\msbf{y}, y))
\, (\eq \otimes \mc{Z}) ( d (\msbf{y}, y), ((x_1, \ldots, x_n), z)) \\
& = \textstyle \int_{X^{(n)}}  \int_Z \tau (-, \msbf{y})
\, \eq  (d \msbf{y}, (x_1, \ldots x_n)) \, \delta (y, z)  \tag*{by FT} \\
& = \textstyle \int_{X^{(n)}}  \int_Z \tau (-, \msbf{y})
\, \delta  (d \msbf{y}, x_1 \cdots x_n) \, \delta (y, z) 
= \tau (-, (x_1 \cdots x_n, z))
\end{align*}

\smallskip


\smallskip

\noindent{(Proof of Claim 1 of (ii))}
It needs to prove for any $n$
and for any $(\msbf{z}, z) \in (X \uplus I)^{(n)} \times Z$, 
\begin{align*}
\tau (-, (G_{n, \infty} \otimes \mc{Z}) (\msbf{z}, z))  = \tau_n (-,(\msbf{z}, z)) 
\end{align*}
As each $\msbf{z} =(1, x_1) \cdots (1, x_{k}) (2, *) \cdots
 (2,*) \in (X \uplus I)^{n}$ for certain $k \leq n$,
the following starts with LHS and ends with RHS.
\begin{align*}
&
\tau(-, (k, (x_1 \cdots x_k, z))) \\
& =
(\pr_l^{(k)} \otimes \mc{Z})
 \Comp \tau_k (-, (x_1 \cdots x_k,z) ) \tag*{by (\ref{genfactotau})} \\
& =  \textstyle 
\int_{(X \uplus I)^{(k)} \times Z} 
\tau_k (-, (\msbf{y},y)) \, (\pr_l^{(k)} \otimes \mc{Z}) (d (\msbf{y}, y), (x_1 \cdots x_k, z)) \\
& =
\textstyle 
\int_{(X \uplus I)^{(k)}} \int_Z 
\tau_k (-, \msbf{y}) \, \delta (d \msbf{y}, (1, x_1) \cdots (1, x_k)) \,
\delta(dy, z)
\tag*{by FT and (\ref{pleft})} \\
& = 
\tau_k (-, ((1, x_1) \cdots (1, x_k), z)) \\
& = \textstyle 
\int_{(X \uplus I)^{(n)} \times Z}  
\tau_n (-, (\msbf{y}, y)) 
\, (p_{n, k} \otimes \mc{Z}) (d (\msbf{y}, y), 
((1, x_1) \cdots (1, x_k),z)) \tag*{
as $\tau_k = (p_{n, k} \otimes \mc{Z}) \Comp \tau_n$} \\
& =
\textstyle 
\int_{(X \uplus I)^{(n)}}  \int_Z
\tau_n (-, (\msbf{y}, y)) 
\, \delta (d \msbf{y}, 
(1, x_1) \cdots (1, x_k) (2, *) \cdots (2, *))
\, \delta(dy, z)
 \tag*{by FT and (\ref{pmn})} \\
& = \tau_n (-, ((1, x_1) \cdots (1, x_k) (2, *) \cdots (2, *), z)) 
\end{align*}

\subsection{\normalsize Accommodating \cite{Crubille}'s Characterisation
 $\dbra{f}{n}$ into the Limit
L$_{\msbf{A}}$ of $\T{\mc{I}}{\opTsKer}$ } \label{accodbra}
\begin{defn}[$\dbra{f}{n}$] \label{defdefdbra}
{\em For an object 
$\msbf{A}=(\mc{A}, \msbf{A}_p)$ in $\T{\mc{I}}{\opTsKer}$
and $f_i \in \msbf{A}_p$ with $i=1, \ldots ,n$,  the morphism
in $\opTsKer$
$$\dbra{f}{n}: \mc{I} \, \, \longrightarrow \, \,  \bigwith_{k \in \mathbb{N}} \mc{A}^{(k)}=\oble{\mc{A}}{\infty}$$
is defined 
by the following instantiation
at each $(k, a_1 \cdots a_k) \in \bigwith_{k \in \mathbb{N}} A^{(k)}$.
\begin{description}
\item[(a)] For $k \leq n$
\begin{align}
 \dbra{f}{n} (I,  (k, a_1 \cdots a_k)) & := {\textstyle \frac{1}{n^k}}
\sum_{\iota : \setone{k} \hookrightarrow  \setone{n}
} \prod_{i=1}^k f_{\iota(i)} (I, a_i) \label{defdbra} \\
& = {\textstyle \frac{1}{n^k}}
\sum_{\iota : \setone{k} \hookrightarrow  \setone{n}
} \prod_{i=1}^k \bar{f}_{\iota(i)} (I, (1, a_i)) 
\tag*{with $\bar{f}_i= f_i \& \mc{I}$} \nonumber \\
& = {\textstyle \frac{1}{n^k}}
\sum_{\iota : \setone{k} \hookrightarrow  \setone{n}
} 
\eq \backslash \bar{f}_{\iota(1)} \otimes \cdots \otimes \bar{f}_{\iota(k)}
(I, (1,a_1) \cdots (1,a_k)) \nonumber \\
& = 
\textstyle \frac{n!}{n^k (n-k)!} \, \, \,  p_{n,k} \Comp \bra{f}{n}
(I, (1, a_1) \cdots (1, a_k)) \tag*{by Lem \ref{pnk}} \nonumber 
\end{align}
That is, in particular letting the above $k$ be $n$,
by the definition of $p_{\infty,n}$ of (\ref{pinf}), 
\begin{align}
 p_{\infty, n} \Comp \dbra{f}{n}  = \frac{n!}{n^k (n-k)!} \bra{f}{n} \, : \,  \mc{I} \longrightarrow
 \oble{\mc{A}}{n} \label{pdbra}
\end{align}
Hence for general $k$ by composing $p_{n,k}$ to (\ref{pdbra}),
as $p_{n,k} \Comp p_{\infty, n}= p_{\infty, k}$,
\begin{align*}
& p_{\infty, k} \Comp \dbra{f}{n}
  = \frac{n!}{n^k (n-k)!} \, \, p_{n.k} \Comp \bra{f}{n} \, : \, \mc{I} \longrightarrow
 \oble{\mc{A}}{n} 
\end{align*}
\item[(b)] For $k > n$, 
$\dbra{f}{n} (I,  (k, a_1 \cdots a_k)) := 0$. 
That is,
$$\pr_k \Comp \dbra{f}{n} := 0 : \mc{I} \longrightarrow \mc{A}^{(k)}.$$
\end{description}
To sum up (a) and (b), the above instantiations at every $k$-th projected
component 
$\mc{A}^{(k)}$ 
yields the unique morphism from $\mc{I}$
to the product
$\underset{k \in \mathbb{N}}{\bigwith} \mc{A}^{(k)}$;
\begin{align*}
& \dbra{f}{n}:= 
 \bigwith_{k \leq n} \frac{n!}{n^{k} (n-k)!} \, \, \pr_l^{(k)} \Comp p_{n,k} \Comp
 \bra{f}{n} 
\, \, \with \, \,  
\bigwith_{k > n} 0, 
\end{align*}
where $\pr_l$ is the left projection $\mc{A} \& \mc{I}  \longrightarrow
 \mc{A}$ and $\bar{f}_i= f_i \& \mc{I}$. 
} \end{defn}
See the following diagram for 
$\dbra{f}{n}$ composed with the $k$-th projection  $\pr_k$ with $k \leq n$:
$$
\xymatrix{ 
\mc{A}^{(k)} \ar@{=}[rr] &    &  \mc{A}^{(k)} \\
\oble{\mc{A}}{k} \ar[u]^{(\pr_l)^{(k)}}
 & \ar[l]_{p_{n,k}} \oble{\mc{A}}{n} &
\ar[u]_{\pr_k} \ar[l]_{p_{\infty, n}}
\underset{k \in \mathbb{N}}{\bigwith} \mc{A}^{(k)} \\ 
&    &    \mc{I} \ar[u]_{\dbra{f}{n}}
\ar[ul]
\ar@{}[ull]^{
\frac{n!}{n^k (n-k)!}
\bra{f}{n}}}
$$

Then $\dbra{f}{n}$ in Definition \ref{defdefdbra}
is shown to provide generators for the homset
$(\oble{\msbf{A}}{\infty})_p$:
\begin{prop}
The double orthogonal homset
\begin{align}
\ccrc{\left(\displaystyle \bigcup_{n \in \mathbb{N}} \{
	     \dbra{f}{n} \mid \forall i \,  f_i \in \msbf{A}_p
       \}\right)} \label{dodbra}
\end{align}
 coincides with the equal subsets in Theorem \ref{charlim}.
\end{prop}
\begin{proof}

\noindent ($(\ref{dodbra}) \subset (iii)$) Straightforward
by Lemma \ref{gencomcc} (ii) with the following inequality
for measurable functions $\mc{I} \longrightarrow \oble{\mc{A}}{\infty}$:
\begin{align}
\dbra{f}{n} \leq 
\textstyle
 \underset{k \in \mathbb{N}}{\bigwith} g^{(k)}
\quad \mbox{with $g= \frac{f_1 + \cdots + f_n}{n}$} \label{dbraam}
\end{align}

\noindent (demonstration of (\ref{dbraam})
As the inequality is point wise, we consider an instantiation 
at $(k, a_1 \cdots a_k)$ for any $k \leq n$; otherwise LHS=0, whereby
the inequality is direct.
\begin{align*}
& \textstyle
( \underset{k \in \mathbb{N}}{\bigwith} g^{(k)})
(I, (k, a_1 \cdots a_k))
=
g^{(k)}  (I, a_1 \cdots a_k)
\\ & = g(I, a_1) \cdots g(I, a_k) 
= (\frac{1}{n})^k
\sum_{j=1}^n f_j (I, a_1) \cdots
\sum_{j=1}^n f_j (I, a_k) \geq  \mbox{Eqn.(\ref{defdbra})}
\end{align*}
\hfill (end of demonstration of (\ref{dbraam}))

\noindent ($(i) \subset (\ref{dodbra})$)
We show
\begin{align*}
 \crc{\left(\displaystyle \bigcup_{n \in \mathbb{N}} \{
	     \dbra{f}{n} \mid \forall   f_i \in \msbf{A}_p
       \} \right) } \subset \crc{((\oble{\msbf{A}}{\infty})_p)}
\end{align*}
For any measure $\nu : \oble{\mc{A}}{\infty} \longrightarrow \mc{I}$,
let $\nu_m : \oble{\mc{A}}{m} \longrightarrow \mc{I}$ be the measure
 in Definition \ref{nun}, definable from $\nu$.
\begin{align*}
 \inpro{\dbra{f}{n}}{\oble{A}{\infty}}{\nu} 
& = \lim_{m \to \infty}
\inpro{\dbra{f}{n}}{\oble{A}{\infty}}{\nu_m \Comp p_{\infty, m}}
\tag*{by Prop \ref{approxnun}} \\
& \geq \inpro{\dbra{f}{n}}{\oble{A}{\infty}}{\nu_n \Comp p_{\infty,
 n}} \tag*{by Ineqn (\ref{unun})}\\
& = 
\inpro{p_{\infty,
 n} \Comp \dbra{f}{n}}{\oble{A}{n}}{\nu_n} \tag*{by recipro.} \\
& = 
 \inpro{\frac{n!}{n^k (n-k)!}   \bra{f}{n}}{\oble{A}{n}}{\nu_n} \tag*{by
 (\ref{pdbra})}\\
& = 
\inpro{\bra{f}{n}}{\oble{A}{n}}{\frac{n!}{n^k (n-k)!} \nu_n} 
\end{align*}
Thus, for $\nu$ belonging to LHS, 
as $\bra{f}{n}$s form the generators of
$(\oble{\msbf{A}}{n})_p$ by Proposition \ref{corgene} (ii),
\begin{align*}
 \frac{n!}{n^k (n-k)!}  \, \nu_n \in \crc{((\oble{\msbf{A}}{n})_p)}
\end{align*}
Take an arbitrary $x_\infty \in \oble{\msbf{A}}{\infty}$,
then we know $x_n= p_{\infty, n} \Comp x_{\infty} \in
 \oble{\msbf{A}}{n}$ for any natural number $n$, hence 
\begin{align*}
1 \geq \inpro{x_n}{\oble{A}{n}}{\frac{n!}{n^k (n-k)!}\nu_n} 
\end{align*}
Equivalently
\begin{align*}
 \frac{n^k (n-k)!}{n!}  \geq 
\inpro{x_n}{\oble{A}{n}}{\nu_n} 
\stackrel{recipro.}{=}  
\inpro{x_\infty}{\oble{A}{\infty}}{\nu_n \Comp p_{\infty, n}}
\end{align*}
Making $n \to \infty$
\begin{align*}
\textstyle 1 = \lim_{n \to \infty} \frac{n^k (n-k)!}{n!}
& \geq 
\lim_{n \to \infty} \inpro{x_\infty}{\oble{A}{\infty}}{\nu_n \Comp
 p_{\infty, n}}  \\
& =
\inpro{x_\infty}{\oble{A}{\infty}}{\nu},
 \tag*{by Prop \ref{approxnun}}
\end{align*}
which means $\nu \in \crc{((\oble{\msbf{A}}{\infty})_p)}$.
\end{proof}

\subsection{\normalsize The Exponential of $\Pcoh$ \cite{DE}} \label{exppcoh}
\noindent On objects: \\ $\bang{X}=( M_{\text{fin}}(\abs{\! X \!},
\Po(\bang{X}))$ is defined for $X=(\abs{\! X \!},\Po{X})$;
\begin{align*}
  P(\bang{X}) := 
  \{ u^! \mid  u \in \Po(X) \}^{\bot \bot}, \quad
\mbox{where} \quad \textstyle  u^!([a_1, \cdots , a_n]) = \prod_{i=1}^n u_{a_i}
\end{align*}

\noindent On morphisms: \\
$\bang{t} \in
\mathbb{R}^{M_{\text{fin}}(I) \times M_{\text{fin}}(J)}$
is defined for given $t \in
\mathbb{R}^{I \times J}$:
\begin{align}
(\bang{t})_{\mu, \nu} & :=
\sum_{\rho \in L(\mu, \nu)}
\frac{\nu!}{\rho!}
t^{\rho}   = \sum_{ \sigma \in \mathfrak{S}_n / S_{\vec{\mu}}}
\, \, \prod_{i=1}^n t_{b_{\sigma(i)},a_{i}} \label{bangt}
 \end{align}

In the first equation of (\ref{bangt}),
for $\mu \in M_{\text{fin}}(I)$ and $\nu \in \in M_{\text{fin}}(J)$, \\
$L(\mu,\nu):=
\left\{ r \in M_{\text{fin}}(I \times J) \, \bigl\lvert \, 
\begin{array}{l}
\forall \, i \in I \sum_{j \in J} r((i,j)) = \mu(i)  \\
\forall \, j \in J \sum_{i \in I} r((i,j)) = \nu(j)
\end{array}  
 \right\}$ \\
\hfill and 
$\nu ! :=
\prod_{j \in J} \nu(j)!$ and 
$\rho ! :=\prod_{(i, j) \in I \times J}
\rho ((i,j))!$

\smallskip

In the second equation of (\ref{bangt}), the multisets are given explicitly by
$\mu=[b_1, \ldots, b_n]$
and $\nu=[a_1, \ldots, a_n]$. 
$\vec{\mu}:=(b_1, \ldots, b_n)$ so that the actions $\sigma$
ranging on the quotient $\mathfrak{S}_n / S_{\vec{\mu}}$
do not depend on the ordering $\vec{\mu}$
of $\mu$.

\begin{rem}
{\em Note the action on objects of $\Pcoh$
may be seen to be subsumed in that on morphisms,
because the definition $u^{!}$ is alternatively given as follows 
in terms of the morphism action (\ref{bangt})
when $u \in \Po{X}$ is identified uniquely as the matrix in 
$\mathbb{R}^{\{ * \} \times \abs{X}}$.
\begin{align*}
u^{!} (\mu):= (\bang{u})_{[*, \ldots, *], \mu}
= \textstyle \prod_{i=1}^n u_{*,a_i}
\end{align*}
In the first equation, 
$[*, \ldots, *]$ is the multiset 
of $\abs{1}=\{ * \}$ whose size is the same as
that of $\mu$.
The second equation holds
because  the stabiliser subgroup $S_{(*, \cdots , *)}=\mathfrak{S}_n$,
 when
putting 
$\mu=[a_1, \ldots, a_n]$. 
}\end{rem}

\end{document}